\documentclass[acmtocl]{acmsmall}

\usepackage[english]{babel}
\usepackage{amssymb}
\usepackage{times}
\usepackage{latexsym}
\usepackage{amsfonts}
\usepackage{amsmath} 
\usepackage{xspace}
\usepackage[all]{xy}
\usepackage{todonotes}
\usepackage{tikz}
\usepackage{enumerate}
\usepackage{intcalc}
\usepackage{array}

\usetikzlibrary{calc}
\usetikzlibrary{shapes}
\usetikzlibrary{decorations.pathreplacing}

\newtheorem{fact}[theorem]{Fact}
\newtheorem{claim}[theorem]{Claim} 

\newcommand{\leqnomode}{\tagsleft@true}
\newcommand{\reqnomode}{\tagsleft@false}

\newcommand{\nc}{\newcommand}

\newcommand{\fA}{\mathfrak{A}}%
\renewcommand{\phi}{\varphi} % Nicer-looking phi

\newcommand{\AAA}{\mbox{\large \boldmath $\alpha$}}
\newcommand{\BBB}{\mbox{\large \boldmath $\beta$}}

\newcommand{\Sat}{\ensuremath{\textit{Sat}}}
\newcommand{\FinSat}{\ensuremath{\textit{FinSat}}}

\newcommand{\FOt}{\mbox{$\mbox{\rm FO}^2$}}
\newcommand{\GF}{\mbox{$\mbox{\rm GF}$}}
\newcommand{\GFt}{\mbox{$\mbox{\rm GF}^2$}}

\newcommand{\GFtTG}{\mbox{$\mbox{\rm GF}^2{+}{\textsc{tG}}$}}
\newcommand{\GFtTRG}{\mbox{$\mbox{\rm GF}^2{+}{\textsc{trG}}$}}
\newcommand{\GFtPG}{\mbox{$\mbox{\rm GF}^2{+}{\textsc{pG}}$}}
\newcommand{\GFtTSG}{\mbox{$\mbox{\rm GF}^2{+}{\textsc{tsG}}$}}
\newcommand{\GFtEG}{\mbox{$\mbox{\rm GF}^2{+}{\textsc{eG}}$}}

\newcommand{\NP}{\textsc{NP}}
\newcommand{\PTime}{\textsc{PTime}}

\newcommand{\ExpTime}{\textsc{ExpTime}}
\newcommand{\ExpSpace}{\textsc{ExpSpace}}
\newcommand{\NExpTime}{\textsc{NExpTime}}
\newcommand{\TwoExpTime}{2\textsc{-ExpTime}}
\newcommand{\TwoNExpTime}{2\textsc{-NExpTime}}

\newcommand{\str}[1]{{\mathfrak{#1}}}
\newcommand{\restr}{\!\!\restriction\!\!}
\newcommand{\N}{{\mathbb N}}   % Natural numbers
\newcommand{\Q}{{\mathbb Q}}   % rationals
\nc{\otheta}{\hat{\theta}}

\nc{\oTheta}{{{\mathit{\hat{\Theta}}}}}  % italic \Theta with a hat
\nc{\nTheta}{{{\mathit{{\Theta}}}}} % \Theta in italic

\nc{\boTheta}{{\bm{\mathit{\hat{\Theta}}}}}  %bold and italic \Theta with a hat
\nc{\bnTheta}{{{\mathit{{\Theta}}}}} % bold and italic normal \Theta

\nc{\uGamma}[2]{\Gamma_{#1}(#2)}

\newcommand{\sss}{\scriptscriptstyle}

\newcommand{\type}[2]{{\rm type}^{{#1}}({#2})}
\newcommand{\ctype}[1]{{\rm ctype}({#1})}
\newcommand{\neat}{neat }
\newcommand{\sizeOf}[1]{|#1|}

\newcommand{\tsplice}[2]{{#1}\restr{#2}}
\newcommand{\Csplice}[2]{{#1}^{#2}}
\newcommand{\ce}{(\exists)}
\newcommand{\cffs}{(\forall \forall_{s})}
\newcommand{\cfs}{(\forall_{s}) }
\newcommand{\cff}{(\forall \forall)}
\newcommand{\cf}{(\forall)}
\newcommand{\cfes}{(\forall \exists_{s})}
\newcommand{\cfe}{(\forall \exists)}

\newcommand{\cutout}[1]{}

\newcommand{\phiuniv}{\phi_{\sss \forall}}
\newcommand{\phisyms}[1]{\phi^{#1}_{{\sss \hspace*{-2pt}\leftrightarrow}}}
\newcommand{\phisym}{\phi^{\phantom{S_i}}_{{\sss \hspace*{-2pt}\leftrightarrow{,}\exists}}}
\newcommand{\phieq}{\phi_{eq}}
\newcommand{\phifull}[1]{\phi^{#1}}

\newcommand{\ung}{\xi}
\newcommand{\bing}{\zeta}

\title{Finite Satisfiability of the Two-Variable Guarded Fragment with Transitive Guards and Related Variants}
\author{Emanuel Kiero\'nski\affil{University of Wroc{\l}aw}  Lidia Tendera\affil{Opole University}}

\category{F.4.1}{Mathematical Logic}{Finite model theory}

\terms{Theory}

\keywords{two-variable logic, guarded fragment, equivalence relation, transitive relation, finite satisfiability problem, computational complexity}

\begin{abstract}
	We consider extensions of the two-variable guarded
	fragment, \GFt{}, where distinguished binary predicates that occur only in guards are required to be interpreted in a special way (as transitive relations, equivalence relations, pre-orders or partial orders). 
	We prove that the only fragment that retains the finite (exponential) model property is \GFt{} with {\em equivalence guards} without equality. For remaining fragments we show that the size of a minimal finite model is at most doubly exponential. 
	To obtain the result we invent a strategy of building finite
	models that are formed from a number of multidimensional grids
	placed over a cylindrical surface.   The construction yields a
	\TwoNExpTime -upper bound on the complexity of the finite satisfiability problem for these fragments. 
	We improve the bounds and obtain optimal ones for all the fragments considered, in particular \NExpTime{} for \GFt{} with equivalence guards, and  \TwoExpTime{} for \GFt{} with {\em transitive guards}.  To obtain our results  we essentially use some results from integer programming.
\end{abstract}

\begin{document}
	\maketitle
	\begin{bottomstuff}
		This work is supported by the Polish National Science Centre grant DEC-2013/09/B/ST6/01535.\\
		Authors' addresses: Emanuel Kiero\'nski, Institute of Computer Science, University of Wroc{\l}aw,
		Joliot-Curie 15, PL-50-383 Wroc{\l}aw, Poland; email: Emanuel.Kieronski@cs.uni.wroc.pl.
		Lidia Tendera, Institute of Mathematics and Informatics, Opole University, 
		Oleska 48, PL-45-052 Opole, Poland; email: tendera@math.uni.opole.pl
	\end{bottomstuff}

	\section{Introduction} 

The two-variable fragment of first-order logic, \FOt{}, and the two-variable guarded fragment, \GFt{},
are widely investigated formalisms whose study is motivated by their close connections to modal, description and temporal logics.
It is well-known that \FOt{} enjoys the finite model property~\cite{Mor75}, and that its satisfiability
($=$ finite satisfiability) problem is \NExpTime-complete~\cite{GKV97}. Since \GFt{} is contained in
\FOt{}, it too has the finite model property; however, its satisfiability problem is slightly easier, namely \ExpTime-complete~\cite{Gra99}. 

It is impossible,  in \FOt{}, to write a formula expressing the condition that
a given binary relation symbol denotes a transitive relation.   This can be shown by giving examples of {\em infinity axioms}, i.e., satisfiable sentences that have only infinite models (see Section \ref{sec:examples}). Transitivity, combined with various additional conditions,  is used in modal logics to restrict the class of Kripke frames, e.g.~to transitive structures for modal logic K4, transitive and reflexive---for S4, or equivalence structures---for S5. In temporal logic various orderings are used to model time flow. Hence, the question  arises as to whether such a feature like transitivity (or related properties like orderings and equivalence relations) could be added at a reasonable computational cost.  

In the past years 
various extensions of \FOt{} and \GFt{} were investigated
in which certain distinguished binary relation symbols are declared to denote transitive relations, equivalence relations, or various  orderings. It turns out that the decidability of these fragments depends on the {\em number} of the distinguished relation symbols available. 
Namely, the following fragments are undecidable for both finite and unrestricted satisfiability:
\GFt{} with two transitive relations \cite{Kie05}, \GFt{} with one transitive relation and one equivalence relation \cite{KMP-HT14}, \GFt{} with three equivalence relations \cite{KO12}, \GFt{} with three linear orders \cite{Kie2011}. 
Decidability of the satisfiability and finite satisfiability problems for many of the remaining cases has been recently established involving several papers not mentioned above  
\cite{Otto01,SchZ12,MZ13,BDM11,KMP-HT14,ZeumeH2016,KP-HT15}.

The above mentioned undecidability results motivated the study of fragments of \GFt{} with distinguished relation symbols, with an additional restriction that the distinguished symbols may appear only in guards~\cite{GMV99}. This restriction comes naturally when embedding many expressive modal logics in first-order logic or when considering branching temporal logics. 
Two notable examples in this category are the two-variable guarded fragment with {\em transitive guards}, \GFtTG{}, and  the two-variable guarded fragment with {\em equivalence guards}, \GFtEG{}, where all distinguished symbols are requested 
to be interpreted as transitive and, respectively, equivalence relations. 
In this case, the satisfiability problem for  \GFtEG{} remains \NExpTime-complete~\cite{Kie05}, while the satisfiability problem for \GFtTG{} is
\TwoExpTime-complete~\cite{ST04,Kie06}, for {\em any} number of special relation symbols used. We remark in passing that allowing as guards {\em conjunctions of transitive guards} leads to an undecidable fragment \cite{Kaz06}. 

The lack of the finite model property for the above mentioned
fragments naturally leads to the question, whether their finite
satisfiability problems are decidable. 
This was addressed in \citeN{ST05} where it is shown that the finite satisfiability problem for \GFtTG{}  with one transitive relation is $\TwoExpTime$-complete, and later in
\citeN{KT07} where  \TwoExpTime\- and \TwoNExpTime\- upper bounds for \GFtEG{} and, respectively,  \GFtTG{} were given.

This paper can be considered as a full, improved and expanded version of \citeN{KT07}. In particular, we give tight complexity bounds for the finite satisfiability problems for the extensions of \GFt{} with {\em special guards}, where some binary relation
symbols are required to be interpreted as transitive relations, pre-orders, partial orders or equivalence relations, but all these {\em 	special} symbols are allowed to appear only in guards. 
The main results of the paper are as follows: (i) \GFtEG{} without equality has the exponential model property and is $\NExpTime$-complete; (ii) the finite satisfiability problem for \GFtEG{} with equality is decidable in $\NExpTime$; (iii) the finite satisfiability problem for \GFtTG{} with equality is decidable in $\TwoExpTime$, cf.~Table~\ref{tabela}. 
The upper bounds given in (ii) and (iii) correspond to the lower bound for \GFtEG{} with one equivalence relation \cite{Kie05}, and the lower bound for \GFt{} with one partial order \cite{Kie06}. 
The above results allow us also to establish tight complexity bounds for the finite satisfiability problems for extensions of \GFt{} with other special guards of the above mentioned form.

\begin{table}
\tbl{Overview of \GFt{} over restricted classes of structures.}
	{\setlength{\extrarowheight}{5pt}
		\begin{tabular}{|c|c|}
			\hline
			{\bf Special symbols} & {\large{\color{white}I}} {\bf Decidability and Complexity}\\
			\hline\hline
			EG without = & {\bf FMP, \NExpTime{}-complete} { this paper}\\
			\hline 
			EG with = & {\bf \NExpTime-complete}, \Sat:\cite{Kie05} {\bf \FinSat}: {this paper} \\
			\hline			

			1 equivalence r. $(*)$ & FMP, \NExpTime-complete \cite{KO12}\\
			\hline
			2 equivalence r. & \TwoExpTime-complete,  \Sat:\cite{Kie05} {\FinSat}: \cite{KP-HT15}\\
			\hline
			3 equivalence r. & undecidable \cite{KO12}\\
		\hline
		TG & {\bf \TwoExpTime{}-complete} \Sat: \cite{ST04} {\bf \FinSat}: { this paper} \\
		\hline
			1 transitive r. & {\bf \TwoExpTime-complete} \Sat: \cite{Kie05} {\bf \FinSat}: { this paper}\\
			\hline
			2 transitive r. & undecidable \cite{Kie05,Kaz06}\\
			\hline
			1 trans. + 1 equiv. & undecidable \cite{KMP-HT14}\\
			\hline

			1 linear order $(*)$& \NExpTime-complete \cite{Otto01}\\
			\hline			
			2 linear orders $(*)$  & Sat: ? \FinSat{} in \TwoNExpTime{}   \cite{ZeumeH2016} \\
			\hline			
			3 linear orders & undecidable \cite{Kie2011}\\
			\hline			
			\end{tabular}} 
		\begin{tabnote}
		\Note{\vspace*{5pt}The upper bounds in the results marked with $(*)$ are actually obtained for full  \FOt{}. The corresponding lower bounds, if present, hold for \GFt{}.}
		\end{tabnote}
	  \label{tabela} 
\end{table}

The related case with linear orderings can be seen as an exception, because the undecidability result for \FOt{} with three linear orders can be adapted to the case when the order relations are used only as guards \cite{Kie2011}. This is, however, not so surprising, as the presence of a linear ordering no longer allows one to construct new models by taking disjoint copies of smaller models---a technique widely applied in guarded logics, and also in this paper.

In this paper we show that every finitely satisfiable
\GFtEG-formula and every finitely satisfiable \GFtTG-formula has a
model of at most double exponential size. 
The restriction of special symbols to guards suggests that one cannot enforce models where an (unordered) pair of two distinct elements belongs to distinct special relations, as e.g.~the formula $\forall xy (Sxy \rightarrow S'xy)$ contains the atom $S'xy$ in a non-guard position. 
Indeed, the constructions of our paper show that every (finitely) satisfiable formula has a {\em ramified} model that satisfies the above-described property. This in particular allows us to build models for \GFtEG-sentences as multidimensional grids (the number of dimensions equals the number of equivalence relations) where intersections of equivalence classes of distinct special relations have at most one element.

The restriction of special symbols to guards does not, however, prevent us from enforcing pairs of elements connected by the same special relation in both direction, i.e., such that $Sab\wedge Sba$ holds (cf.~Example \ref{ex:large-no-equality} in Section \ref{sec:examples}; see also the proof of Lemma 2 in
\citeN{Kie05}, where it is shown
how to enforce in \GFtTG{} a transitive relation to be an equivalence).  
Hence, when constructing a model for a \GFtTG-formula we first construct a model for the {\em symmetric part} of the formula that have the above described properties of
models for \GFtEG-formulas.         
And then, we place a number of such models over a cylindrical surface and connect them in a regular way, using only non-symmetric special connections. This way we obtain a structure satisfying both the {\em non-symmetric} and the symmetric parts of our formula. In case when the special symbols are required to be partial orders, the construction simplifies, as the symmetric part of the formula is essentially empty. We refer the reader to Section \ref{sec:normalforms}, where the intuition concerning symmetric and non-symmetric parts of the formula is formalized.

Our results give also a trivial, double exponential upper bound on the
size of a single equivalence class in a finite model of a
\GFtEG-formula. We argue that this bound is optimal,
which  contrasts with the case of general satisfiability, where it can be shown that
every satisfiable \GFtEG-sentence has a (possibly infinite) model
with at most exponential equivalence classes.

The established bounds on the size of minimal models  yield  \TwoNExpTime -upper bounds on the complexity of
the finite satisfiability problem for these fragments. 
We work further and improve the bound for \GFtEG{} to \NExpTime, and for \GFtTG{}  to \TwoExpTime; as discussed above, 
these bounds are optimal. To obtain them we do not work on the level of 
individual elements, but rather on the level of equivalence classes of models, or, more precisely,  on their succinct representations in the form of {\em counting types} (cf.~Definition~\ref{d:counting-type} in Section~\ref{sec:prelim}) that give the number of realisations of each $1$-type 
in a class. In fact, we first show that it suffices to count the number of realizations of $1$-types up to some fixed number depending only on the size of the signature, so that the number of all relevant counting types becomes bounded as well.

It is perhaps worth mentioning that there is another interesting
extension of the guarded fragment, namely \GF{} with fixpoints,
that has been shown decidable for satisfiability \cite{GW99} and
of the same complexity as pure \GF{}. Decidability of the
finite satisfiability problem for this fragment and exact complexity bounds have been recently shown  by
\citeN{BaranyB12}.

We also remark that tight exponential
complexity bounds for \GFt{} with counting quantifiers for both
satisfiability and finite satisfiability have been established by \citeN{PH07}, and the interplay between transitive guards and counting for unrestricted satisfiability has been studied by \citeN{T05}. 

The structure of the article is as follows. In Section~\ref{sec:prelim} we introduce various extensions of \GFt{} with special guards and formulate a `Scott-type'  normal form for formulas in these logics. There, in Subsection~\ref{sec:examples}, we also present several examples showing the expressive power of the various fragments. In Section~\ref{sec:types-and-counting} we introduce the main technical notions and give first important observation concerning counting types based on properties of \FOt{}. In Section~\ref{sec:eq-no-equality} we show that \GFtEG{} without equality has the finite model property. In Section \ref{sec:eq-with-equality} we proceed to \GFtEG{} with equality, where we prove the main model theoretic and complexity results concerning \GFtEG{}, and also prepare for the case of \GFt{} with transitive guards studied in Section \ref{sec:transitive}. We conclude with a discussion on the complexity of related extensions and list a few open questions. 

{\bf Related work.} 
There are natural decidable fragments beyond \GF{}. They often allow one to use {\em conjunctions of guard atoms} instead of atomic guards.
Since arbitrary conjunctions of guards lead easily to undecidability, additional restrictions are imposed  on the conjunctions to retain decidability (cf.~loosely guarded \cite{Ben97}, clique guarded \cite{Gradel99}, packed fragment \cite{Marx01}). In the case with two-variables these restrictions boil down to disallowing guards of the form $Px\wedge Qy$ expressing the cross product of two unary relations. The impact of using cross products in \GF{} has been recently studied by \citeN{BourhisMPieris17}.

Transitivity has been observed to easily lead to undecidability both in the description logics world and in database-inspired reasoning problems. Therefore, various syntactic restrictions are introduced to ensure decidability; the restriction to special guards as studied in this paper is one possible way.

The standard description logic allowing one to express transitivity of roles is the DL logic $\mathcal S$ that enjoys the finite model property, even when extended with inverses of roles. Finite model property is lost when additionally roles hierarchies and number restrictions are added, like in $\mathcal{SHIQ}$
	(see, e.g.,~\citeN{Tobies01}). The combination of role hierarchies and transitive roles is expressible in \GFt{} with transitive relations but not directly expressible in \GFtTG{}. The logic $\mathcal{SHIQ}$ is subsumed by even more expressive DL logic $\mathcal{SROIQ}$, where roles can be declared transitive, (ir)reflexive, (anti)symmetric or disjoint. 
(Finite) satisfiability of $\mathcal{SROIQ}$ has been shown decidable  and \TwoNExpTime-complete by \citeN{Kaz08} by an exponential-time reduction to the case of the two-variable fragment with counting. 

Compatibility of transitivity with {\em existential rules}, aka Datalog$\pm$, has been recently studied, e.g.,~by \citeN{GPT13}, \citeN{BagetBMR15} and \citeN{AmarilliBBB16}, who showed in particular how fragile some of the conditions are and how easily slight changes might lead to undecidability.

We also point out that even though in this paper we discuss fragments for which it required more care to prove decidability of the finite satisfiability problem than to prove decidability of the satisfiability problem, this by no means is a general rule. In particular, there are logics such that \Sat{($\cal{L}$)} is undecidable and \FinSat{($\cal{L}$)} is decidable, or vice versa (see, e.g.,~\citeN{MichOW12} for a family of examples from the elementary modal logics).

\section{Preliminaries}\label{sec:prelim}
We employ standard terminology and notation from model theory
throughout this paper (see, e.g.,~\citeN{ck}). In particular, we refer
to structures using Gothic capital letters, and their domains using
the corresponding Roman capitals.  

We denote by \GFt{} the guarded
two-variable fragment of first-order logic (with equality), without
loss of generality restricting attention to signatures of unary and
binary relation symbols (cf.~\citeN{GKV97}).
Formally,
\GFt{} is the
intersection of \FOt{} (i.e.,~the restriction of first-order logic in which only two variables, $x$ and $y$ are available)
and the \emph{guarded fragment}, \GF{}~\cite{ABN98}.
\GF{} is defined as the least set of formulas such that:
(i) every atomic formula belongs to \GF{};
(ii) \GF{} is closed under logical connectives $\neg, \vee,
\wedge, \rightarrow$;
and (iii) quantifiers are appropriately relativised by atoms.
More specifically, in \GFt{}, condition (iii) is understood as follows:
if $\phi$ is a formula of \GFt{},
$\bing$ is an atomic formula 
containing all the free variables of $\varphi$, and $u$ (either $x$ or $y$) is a free variable in $\bing$, then the formulas
${\forall} {u}(\bing \rightarrow  \phi)$ and
${\exists}  {u}(\bing \wedge \phi )$ belong to \GFt{}.
In this context, the atom $\bing$ is called a {\em guard}.
The equality symbol $=$ is allowed in guards.
We take the liberty of counting as guarded those formulas which can be made guarded by trivial logical manipulations, e.g.~$\forall x\forall y (Ryx \wedge \phi \rightarrow \psi)\in \GFt$, for any binary predicate symbol $R$ in the signature and any $\phi$, $\psi \in \GFt$.

\subsection{Special Guards} \label{sec:specialguards}

In this paper we work with relational signatures containing a distinguished subset of {\em special} binary relation symbols, often denoted by $S, S_1, S_2, \ldots$, required to be interpreted in a {\em special way} (as transitive relations,  transitive and reflexive relations, equivalences or partial orders). 
For a given signature $\sigma$, a $\sigma$-structure is {\em special}, if all the special relation
symbols are interpreted in the required way. A \GFt-formula $\phi$ with special  relation symbols  is (finitely) satisfiable, if there exists a special (finite) model of $\phi$.

The two-variable guarded fragment with {\em special guards}
is the extension of \GFt{} where the {\em special} symbols are allowed to appear only in guards. In particular, formulas of the form $\forall x \forall y (Sxy\rightarrow S'xy)$ are not allowed.
\footnote{For technical reasons we later also allow special guards of the form $Sxy \wedge Syx$, $Sxy \wedge \neg Syx$, or $Syx \wedge \neg Sxy$ (cf.~Discussion after Definition~\ref{def:normal}).}
Since the computational complexity of the (finite) satisfiability problem for \GFt{} with special guards depends on the properties required for the special symbols, we distinguish more specific fragments indicating the properties 
required from the special relations.

\GFtTG{}: all special relations are required to be interpreted as transitive relations; this fragment is usually called \GFt{} with {\em transitive guards}.

\GFtTRG{}: all special relations are required to be pre-orders, i.e., transitive and reflexive, 

\GFtTSG{}: all special relations are required to be partial equivalences, i.e., transitive and symmetric,

\GFtEG{}: all special relations are required to be equivalences,

\GFtPG{}: all special relations are required to be partial orders, i.e,~transitive, reflexive and weakly antisymmetric.

Obviously, reflexivity of a special binary relation $T$ can be expressed using equality in guards:
$\forall x (x=x \rightarrow \exists y (Txy \wedge x=y))$
(note that the atom $Txy$ occurs as a guard). Similarly, a formula 
$\forall x\forall y (Txy \wedge Tyx \rightarrow x=y)$ expressing that $T$ is weakly antisymmetric can be simulated using another (non-special) binary relation $R$ and writing a conjunction $\forall x\forall y (Txy\rightarrow Rxy) \wedge \forall x\forall y (Rxy \rightarrow (\neg Ryx \vee x=y))$. 

Taking the above observations into account we have the following inclusions
$$\GFtEG \subseteq \GFtTSG \quad\mbox{ and }\quad \GFtPG \subseteq \GFtTRG \subseteq \GFtTG.$$

On the other hand a \GFt-formula $\forall x \forall y (Txy \rightarrow Tyx)$ expressing that a binary relation $T$ is symmetric cannot be used to define symmetry when $T$ is a special symbol, because the atom $Tyx$ occurs outside guards. 
However,  there is a simple linear reduction  from \GFtEG{} to \GFtTG{} preserving both finite and unrestricted satisfiability \cite{Kie05} (Lemma~2, p.~315). The reduction can be easily modified to work from \GFtTSG. 
We observe additionally that \GFtTG{} reduces to \GFtTRG{}, and that the same reduction works from \GFtTSG{} to \GFtEG. 

\begin{lemma}\label{l:reduction}
	There is a linear time reduction transforming any \GFtTG{} sentence $\phi$ over a signature $\sigma$ into  a \GFtTRG{}-sentence $\phi'$ over a signature $\sigma'$ such that $\phi$ and $\phi'$ are satisfiable over the same domains  and $\sigma'$ consists of $\sigma$ together with some additional unary and (non-special) binary relation symbols.

\end{lemma}

\begin{proof}
	Let $\phi$ be \GFtTG{} formula over $\sigma$ and let $T_1,\ldots,T_k$ be the transitive relation symbols in $\sigma$. 
	We translate $\phi$ to $\phi'$ over a signature $\sigma'$ consisting of $\sigma$ together with fresh binary relation symbols $R_1,\ldots,R_k$ and fresh unary relation symbols
	$U_1,\ldots,U_k$. Let $\chi$ be the conjunction of all formulas of the following form, for every $1\leq i \leq k$,
	\begin{eqnarray*}
& 	\forall xy (T_ixy \rightarrow R_ixy)\\
&	 \forall x (x=x\rightarrow (\neg U_ix \rightarrow \forall y (R_ixy \rightarrow (\neg R_iyx \vee x=y))))
\end{eqnarray*}
	Let $\bar{\phi}$ be the formula  obtained from $\phi$  by replacing, for every $i$ ($1\leq i\leq k$), every atom $T_iuu$ for $u \in \{x, y \}$ by the formula $(T_iuu\wedge U_iu)$ and every binary atom $\bing(x,y)$ containing $T_i$ by the formula $(\bing(x,y) \wedge (x\neq y \vee U_i x))$. 
	The formula $\phi'$ is a conjunction of $\chi\wedge \bar{\phi}$. Note that $\phi'$ is guarded and the transitive relation symbols appear only as guards. 
	
	Now, let $\str{A}\models \phi$. We define $\str{A}'$, a model of $\phi'$, taking $\str{A}$, keeping the interpretation of all non-special symbols, adding reflexive pairs to the interpretation of the transitive relation symbols: 
	$$T^{\str{A}'}=T^{\str{A}}\cup \{(a,a) \mid a\in A\}$$	
	defining $R_i^{\str{A}'}=T_i^{\str{A}'}$ and setting the interpretation of each $U_i$ as the set of elements on which $T_i$ is reflexive in $\str{A}$: $$U_i^{\str{A'}}=\{a\in A \mid \str{A}\models T_iaa\}.$$ 
	It is readily verified that $\str{A}'\models \phi'$ and every $T_i$ is reflexive and transitive in $\str{A}'$. 
	
	Similarly, let $\str{A}'$ be a model of $\phi'$ where every $T_i$ is reflexive and transitive. Define $\str{A}$ as the reduct of $\str{A}'$ to $\sigma$, where additionally the  interpretation of the transitive relation symbols contains a reflexive pair of elements only if they were marked by $U_i$ in $\str{A}'$. Formally, 
	$$T_i^{\str{A}}=T_i^{\str{A}'}\setminus \{(a,a) \mid \str{A}'\models \neg U_ia\}.$$
	We show that  after this operation $T_i$ remains transitive. Suppose $\str{A}\models T_iab\wedge T_ibc \wedge \neg T_iac$ for some $a,b,c\in A$.  Since $T_i$ is transitive in $\str{A}'$ we have  $\str{A}'\models T_iab\wedge T_ibc \wedge T_iac$. Defining the interpretation of $T_i$ in $\str{A}$ we only removed some reflexive pairs from $T_i^{\str{A}'}$, hence we have $a=c$ and $\str{A}'\models \neg U_ia$. 
  Then since $\str{A}'\models \chi$ we have that $\str{A}'\models R_iab\wedge R_iba$ and, hence, $a=b$. But this implies $\str{A}'\models T_ibc \wedge \neg T_ibc$, a contradiction. 
  
	One can verify that indeed $\str{A}\models \phi$. 
	
	The same reduction applies when $\phi\in \GFtTSG$; then $\phi'\in \GFtEG$. 
\end{proof}

\subsection{Examples}\label{sec:examples}

Below we present a few examples demonstrating the expressive power of \GFt{} with equivalence guards or with transitive guards.
To justify our study over finite structures we present satisfiable formulas having only infinite models. We also  illustrate
some difficulties arising when dealing with finite models, in particular we present formulas whose finite models have at least doubly 
exponentially many elements.

\begin{example}[Enforcing infinite models in \GFtEG{} with equality]\label{ex:infiniteEG}
If we restrict the number of equivalence symbols to one, then the
finite satisfiability coincides with satisfiability, since even whole \FOt{} with one equivalence relation
has the finite (exponential) model property \cite{KO12}.
However, in the presence of two equivalence symbols, there are satisfiable \GFtEG{} formulas whose all models
are infinite. Let us recall a simple example of such a formula from \citeN{KO12}. Consider the conjunction $\lambda$ of the
following formulas:
\begin{eqnarray*}
& \exists x \; (Px \wedge Sx) \wedge \forall x \; (Sx \rightarrow \neg \exists y \; (E_2xy \wedge x \not= y  ))\\
& \forall x \; (Px \rightarrow \exists y \; (E_1xy \wedge x \not= y \wedge Qy)) \wedge \forall x \; (Qx \rightarrow \exists y \; (E_2xy \wedge x \not= y \wedge Py))\\
&\forall xy \; (E_1xy \rightarrow ((Px \wedge Py) \rightarrow x=y   )) \wedge \forall xy \; (E_2xy \rightarrow ((Qx \wedge Qy) \rightarrow x=y   ))
\end{eqnarray*}
The formula says that  $S \cap P \not= \emptyset$; the $E_2$-class of any element of $S$ is trivial (a singleton);
every element of $P$ is $E_1$-equivalent to one in $Q$; every element of $Q$ is $E_2$-equivalent to one  in $P$;
each $E_1$-class contains at most one element from $P$; each $E_2$-class contains at most
one element from $Q$.

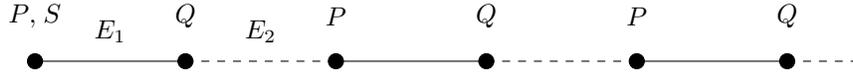
\begin{figure}[htb]
\begin{center}
\begin{tikzpicture}[scale=2]

\foreach \x in {1,2,3,4,5, 6}
   \foreach \y in {0}
      \filldraw[fill=black] (\x, \y) circle (0.05);  
			
\foreach \x in {3,5}{
   \draw (\x,0) -- (\x +1, 0);  			
	\coordinate [label=center:$P$] (A) at ($(\x,0.3)$); 			
 }

\coordinate [label=center: $P{,}  \;S$] (A) at ($(1,0.3)$); 		
\draw (1,0) -- (2,0);

\foreach \x in {2,4}{
   \draw[dashed] (\x,0) -- (\x +1, 0);  			
	\coordinate [label=center:$Q$] (A) at ($(\x,0.3)$); 			
 }

\coordinate [label=center:$Q$] (A) at ($(6,0.3)$); 			
\draw[dashed] (6,0) -- (6.5, 0);

 \coordinate [label=center:$E_1$] (A) at ($(1.5,0.2)$); 			
\coordinate [label=center:$E_2$] (A) at ($(2.5,0.2)$);

\end{tikzpicture}
\caption{An infinite model of $\lambda$.} \label{f:infaxiom}
\end{center}
\end{figure}

It is easy to see that the infinite chain depicted in Figure~\ref{f:infaxiom}  on which $P$ and $Q$ alternate is a model of $\lambda$ and
that every model of $\lambda$ must embed such an infinite chain as a substructure. 
Indeed, the $\forall \exists$-conjuncts always
require fresh elements as witnesses, as an attempt of reusing one of the earlier elements leads to a violation of one of the $\forall \forall$-conjuncts
(or the conjunct stating that the $E_2$-classes of elements in $S$ are singletons). 
\end{example}

\begin{example}[Enforcing large finite models in \GFtEG{} with equality] \label{klasa}
Let us now observe, that finite models for \GFtEG{} have different
properties from infinite models. In the following example we show
how to construct a family of finitely satisfiable  formulas
$\{\lambda_n\}_{n \in \N}$ with only two equivalence relation symbols, such that
every finite model of $\lambda_n$ contains at least one equivalence
class of size at least doubly exponential in $n$, and $|\lambda_n|$
is polynomial in $n$. This is in contrast to the (unrestricted)
satisfiability: as we mentioned,  \citeN{Kie05}  observed that every
satisfiable \GFtEG-formula $\phi$ has a model, in which all
equivalence classes have size at most exponential in $|\phi|$. The
following example is a refinement of the example from 
\citeN{ST05} which used several transitive relations.

Let us assume that $E_1$ and $E_2$ have to be interpreted as
equivalence relations. We construct a finitely satisfiable formula
$\lambda_n$, such that its every model contains some number of full
binary trees of depth $2^n$, whose every leaf requires a root in
its $E_2$-class. Trees will have to be disjoint, so there will
always be $2^{2^n}$ leaves per one root, which will guarantee the
existence of at least one large $E_2$-class.

Except for $E_1$ and $E_2$, we use only a number of unary relation symbols.
Symbols $P_0, \ldots, P_{n-1}$ encode in each element a number from
$\{0, \ldots 2^n-1 \}$ which is the depth of the element in the
tree (the number of the level to which the element belongs). Let
us denote by ${\cal L}_i$ the $i$-th level, i.e., the set of
elements $a$, with the encoded number $i$. Symbol $R$ indicates
roots, symbol $L$ is used to distinguish the left successors from right
successors. The formula $\lambda_n$ consists of conjuncts expressing the following conditions (cf.~Figure~\ref{fig:double-class}):

\begin{enumerate}\itemsep0pt
\item There exists an element satisfying $R$.
\item Every element satisfying $R$ belongs to ${\cal L}_0$.
\item For each element $a$ from ${\cal L}_{2i}$ there exist two elements in ${\cal L}_{2i+1}$ connected to $a$ by $E_1$. One of them
satisfies $L$, the other $\neg L$.
\item For each element $a$ from ${\cal L}_{2i+1}$ (for $2i+1 < 2^n-1$) there exist two elements in ${\cal L}_{2i+2}$ connected to $a$ by $E_2$. One of them satisfies $L$, the other $\neg L$.

\item Every element in ${\cal L}_{2^n - 1}$ is connected by $E_2$ to an element satisfying $R$.

\item If two distinct elements belong to a level ${\cal L}_{2i}$ for some $i$, then they are not connected by
$E_1$.
\item If two distinct elements belong to a level ${\cal L}_{2i+1}$ (for $2i+1<2^n-1$), then they are not connected by $E_2$.
\end{enumerate}
 
 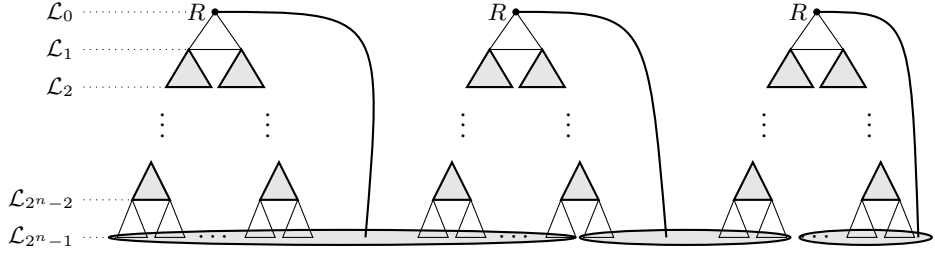
\begin{figure}[hbt]
 	\begin{center}
 		\begin{tikzpicture}[scale=1.0]
 	
 		\draw[style=thick,fill=gray!20] (3.44,0) ellipse(3.1cm and 0.1cm);
 		\draw[style=thick,fill=gray!20] (8.0,0) ellipse(1.4cm and 0.1cm);
 		\draw[style=thick,fill=gray!20] (10.4,0) ellipse(0.88cm and 0.1cm);
 						
 \draw[style=thick] (3.75,0).. controls (4.0,3).. (1.75,3);
 \draw[style=thick] (7.75,0).. controls (7.5,3).. (5.75,3);
 \draw[style=thick] (11.1,0).. controls (11.0,3).. (9.75,3);

 		\draw[dotted,-] (0,3) -- (1.5,3);
 		\draw (0,3) node[left] {\small  ${\cal L}_0$};
 		\draw[dotted,-] (0,2.5) -- (1.5,2.5);
 		\draw (0,2.5) node[left] {\small ${\cal L}_1$};
 		\draw[dotted,-] (0,2.0) -- (1.5,2.);
 		\draw (0,2.0) node[left] {\small ${\cal L}_2$};
 		
 		\draw[dotted,-] (0,0.5) -- (1.0,0.5);
 		\draw (0,0.5) node[left] {\small  ${\cal L}_{2^n-2}$};
 		\draw[dotted,-] (0,0) -- (0.5,0);
 		\draw (0,0) node[left] {\small  ${\cal L}_{2^n-1}$};

 		\foreach \x in {1,5,9}
	 		{\draw (0.4+\x,2.5)  -- (1.1+\x,2.5)  -- (0.75+\x, 3) -- cycle;
	 		\draw (\x+0.5,3) node {\small $R$};
 			\filldraw (\x+0.75, 3) circle (0.04);
		 	\draw (\x+0.75,0) node {$\cdots$};
		 	\draw (\x+0.05,1.6) node {$\vdots$}; 
		 	\draw (\x+1.45,1.6) node {$\vdots$};
		 			 	} 		
 		
		\foreach \x in {1,5,9}
 		{\foreach \y in {0}
 			{\draw[style=thick,fill=gray!20] (0.4+\x+\y,2.5)  -- (0.1+\x+\y,2)  -- (0.7+\x+\y, 2) -- cycle;
 				\draw[style=thick,fill=gray!20] (1.1+\x+\y,2.5)  -- (0.8+\x+\y,2)  -- (1.4+\x+\y, 2) -- cycle;
 			}
 		}

 		\foreach \x in {0.15,1.85,4.15,5.85,8.15,9.85}
 		\filldraw[style=thick,fill=gray!20] (0.5+\x,0.5)  -- (1+\x,0.5)  -- (0.75+\x, 1) -- cycle;
 	
 		\foreach \x in {0.15,1.85,4.15,5.85,8.15,9.85}
 		{\foreach \y in {0}
	  		{\draw (0.5+\x+\y,0.5)  -- (0.3+\x+\y,0)  -- (0.7+\x+\y, 0) -- cycle;
 		  	\draw (1+\x+\y,0.5)  -- (0.8+\x+\y,0)  -- (1.2+\x+\y, 0) -- cycle;
 		  	}
 		}

 		\end{tikzpicture}
 	\end{center}
 	\caption{Doubly exponential $E_2$-class: thin lines correspond to $E_1$ connections, thick lines---to $E_2$; grey regions depict $E_2$-classes.}
 	\label{fig:double-class}
 \end{figure}

It is not difficult to formulate the above sentences in \GFtEG{}, and to see, that each $\str{A} \models \lambda_n$ has a
desired large $E_2$-class. Indeed, suppose $\str{A}\models \lambda_n$. By (1) and (2) there is an element $a\in A$ on level 0 such that $\str{A}\models Ra$. Then, (3) implies that there exist $b, b'\in A$ on level $1$ such that  $\str{A}\models E_1ab \wedge Lb \wedge E_1ab' \wedge \neg Lb'$. Hence, $b\neq b'$ and $\str{A}\models E_1bb'$. Moreover, by (7) $\str{A}\models \neg E_2bb'$.  Now, by (4) there exists  $c,c'\in A$ such that $\str{A}\models E_2bc \wedge Lc \wedge E_2bc' \wedge \neg Lc'$ and there exists $d,d'\in A$ such that $\str{A}\models E_2b'd \wedge Ld \wedge E_2b'd' \wedge \neg Ld'$. Again, we have $c\neq c'$, $d\neq d'$
and $\str{A}\models E_2cc' \wedge E_2dd'$.  
Note that $c\neq d$ and $c'\neq d'$ as $b$ and $b'$ are not in the same $E_2$-class by (7). Repeating the argument for elements at next levels one can show that each root element induces a binary tree with doubly exponentially many leaf nodes (elements on level $2^n-1$).  Moreover, trees induced by different root elements do not overlap. Indeed, suppose $v, v'$ are distinct elements on, say, level $2i$ and $u$ is an element on level $2i+1$ such that $\str{A}\models E_1vu \wedge E_1v'u$. Then, $\str{A} \models E_1vv'$ but this is a contradiction with (6). In case $v,v'$ are on level $2i+1$ we get a contradiction with (7).  

Now, condition (5) ensures that all leaf nodes are connected by $E_2$ to some root element. Hence they are partitioned into at most $k$ $E_2$-classes, where $k$ is the number of root elements. Since the number  of leaf nodes is at least  $k 2^{2^n-1}$, some of the $E_2$-classes have at least $2^{2^n-1}$ elements.   

Observe that to express the last two properties in the example above we need the equality symbol. 
As we will see in Section \ref{sec:eq-no-equality} this is crucial.
\end{example}

\begin{example}[Enforcing infinite models in \GFtTG{} or \GFtTRG{} without equality] \label{ex:infinite-no-equality}
In \GFtTG{} we can easily enforce infinite models even without equality, by the following formula saying that every element in $P$ has a
$T$-successor in $P$ and $T$ is irreflexive.
$$\exists x Px \wedge \forall x (Px \rightarrow \exists y (Txy \wedge Py)) \wedge \neg \exists x Txx.$$
The same task is only slightly harder in \GFtTRG{}. We construct $\lambda^*$ as the conjunction of the following formulas, using an auxiliary
unary symbol $U$ and a binary, non-special symbol $R$:
\begin{eqnarray*}
&\exists x (Px \wedge Ux) \wedge \forall x (Px \rightarrow \exists  y (Txy \wedge Py \wedge (Ux \leftrightarrow \neg Uy)))\\
& \forall xy (Txy \rightarrow Rxy \vee (Ux \leftrightarrow Uy))\\
&\forall xy (Rxy \rightarrow \neg Ryx)
\end{eqnarray*}
This formula is satisfied in the model whose universe consists of natural numbers, $P$ is true everywhere, $U$ is true precisely at even elements,
$T$ is the \emph{less-than-or-equal} relation and $R$ is the \emph{strictly-less} relation. 

Now, let $\str{A}\models \lambda^*$. 
The first conjunct says that $P\cap U \neq \emptyset$ and that each element  in $P$ has another element in $P$ connected by $T$ such that only one of them is in $U$. So, we have at least two elements $a_0, a_1\in A$ such that $\str{A}\models Ua_0\wedge \neg Ua_1 \wedge Ta_0a_1$.  The second conjunct ensures $\str{A}\models Ra_0a_1$. 
The last two conjuncts allow for a symmetric $T$-connection between two elements only when they do not differ with respect to $U$. Hence, $\str{A}\models \neg Ta_1a_0$ and $a_1$ has a new witness $a_2$ for the first conjunct. It follows that $\str{A}\models Ua_2 \wedge Ta_1a_2 \wedge \neg Ta_2a_1$. Repeating the argument one can see that no element from the $T$-chain  $a_0, a_1, \ldots, a_k$ constructed as above can be used as a witness of the first conjunct  for the element $a_k$. Hence, all models of $\lambda^*$ must be infinite. 
\end{example}

\begin{example}[Enforcing large finite models in \GFtTG{} or in \GFtTRG{}, without equality]\label{ex:large-no-equality}
We now explain how to construct a family $\{ \lambda_n^* \}_{n \in \N}$ of finitely satisfiable  \GFtTG{} formulas without equality
employing one transitive symbol $T$ such that every model of $\lambda^*_n$ contains $2^{2^n}$ $T$-cliques (for a formal definition
of a $T$-clique see the beginning of Section \ref{sec:types-and-counting}) 
forming an alternating pattern depicted in Figure~\ref{f:altchain}.  
Since our intended models are reflexive, we can also treat $\lambda_n^*$ 
as $\GFtTRG{}$ formulas.

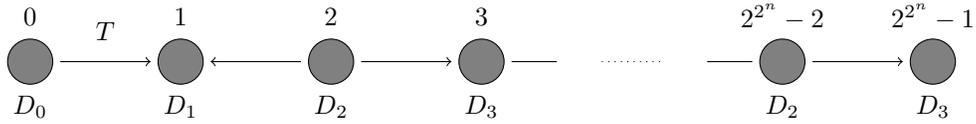
\begin{figure}[htb]
\begin{center}
\begin{tikzpicture}[scale=2]

\foreach \x in {1,2,3,4, 6,7}
   \foreach \y in {0}{
      \filldraw[fill=gray] (\x, \y) circle (0.15);  
				}
			
\draw[->] (1.2,0) -- (1.8,0);
\draw[<-] (2.2,0) -- (2.8,0);
\draw[->] (3.2,0) -- (3.8,0);
\draw[->] (6.2,0) -- (6.8,0);
\draw[-] (4.2,0) -- (4.5,0);
\draw[-] (5.5,0) -- (5.8,0);
\draw[dotted] (4.8,0) -- (5.2,0);
			
\foreach \x in {0,1,2,3}{
     			\coordinate [label=center:$\x$] (A) at ($(\x+1,0.3)$); 			
					\coordinate [label=center:$D_{\x}$] (A) at ($(\x+1,-0.3)$); 			
					}
			
\coordinate [label=center:$2^{2^n}-2$] (A) at ($(6,0.3)$); 		
\coordinate [label=center:$2^{2^n}-1$] (A) at ($(7,0.3)$); 		

\coordinate [label=center:$D_{2}$] (A) at ($(6,-0.3)$); 		
\coordinate [label=center:$D_3$] (A) at ($(7,-0.3)$); 		

\coordinate [label=center:$T$] (A) at ($(1.5,0.2)$);

\end{tikzpicture}
\caption{An alternating chain of $T$-cliques} \label{f:altchain}
\end{center}
\end{figure}

We partition all elements into four sets distinguished by unary predicates $D_0, \ldots, D_3$.
This way each clique is able to distinguish its successor clique from its predecessor clique.
We employ also unary predicates $P_0, \ldots, P_{n-1}$ and think that for every element they 
encode its \emph{local coordinate} in the range $[0, 2^n-1]$. In a standard fashion we can write
the following quantifier-free formulas:
\begin{itemize}
\item $succ(x,y)$, saying that the local coordinate of $y$ is greater by $1$ (addition modulo $2^n$)
than the local coordinate of $x$, 
\item $less(x,y)$, saying that the local coordinate of $x$ is smaller 
than the local coordinate of $y$, and 
\item $equal(x,y)$, saying that the local coordinates of $x$ and $y$ are equal.
\end{itemize}

We enforce the existence of the cliques by the following formulas:
\begin{eqnarray*}
&\forall x (D_l x \rightarrow \exists y (Tyx \wedge D_ly \wedge succ(y,x))) \qquad \mbox{(for even $l$)}\\
&\forall x (D_l x \rightarrow \exists y (Txy \wedge D_ly \wedge succ(x,y))) \qquad \mbox{(for odd $l$).} 
\end{eqnarray*}
They guarantee 
that in any model, starting from any element in $D_l$ we can
find in $D_l$ a $T$-chain of elements with all possible local coordinates. As we work with finite models such a chain must eventually form 
a loop on which all possible local coordinates appear. 

We endow each $T$-clique so obtained with a \emph{global coordinate} in the range $[0, 2^{2^n}-1]$ by regarding
its elements as indices of binary digits encoded by an additional unary predicate $B$. 
We assume that every element with local coordinate $i$ encodes the $i$-th bit of such global coordinate.
The
consistency of this encoding can be enforced in a natural way, by writing for all $l$:
\begin{eqnarray*}
&\forall xy (Txy \rightarrow (\bigwedge_l (D_lx \leftrightarrow D_ly) \wedge \bigwedge_i( P_ix \leftrightarrow P_iy)) \rightarrow (Bx \leftrightarrow By)).
\end{eqnarray*}
 
We further say that there exists a clique in $D_0$ with global coordinate $0$: 
\begin{eqnarray*}
\exists x (D_0x \wedge \forall y (Tyx \rightarrow (D_0y \rightarrow \neg By))),
\end{eqnarray*}
and that each element with local coordinate $0$:
in a clique of type $D_l$ with global coordinate $k < 2^{2^n}$ is joined by a $T$-edge,
oriented right if $k$ is even and oriented left if $k$ is odd, to an element with local coordinate $0$ in $D_{l+1 (\text{mod } 4)}$.

This alternating pattern of $T$-connections provides guards precisely for the pairs of elements 
belonging either to the same clique or to two consecutive cliques. This allows us now to  complete $\lambda_n^*$ by saying that each clique has global 
coordinate greater by one than its predecessor clique.
This can be done as in  the proof of Theorem~7.1 from \citeN{KMP-HT14}
where a similar counting to $2^{2^n}$ is organized in \GFt{} with two equivalence relations (not restricting the equivalence symbols
to guard positions) using non-trivial intersections of equivalence classes. For the sake of completeness we give some details below.

We take the predicate $B^{\sss 1}$ to mark in each clique the least significant position satisfying $B$, and
we take  $B^{\sss 0}$ to mark the least significant positions not satisfying $B$. To this end we write for all $l$ the following \GFtTG{} formulas:
\begin{eqnarray*}
& \forall x \big(B^{\sss 1}x \leftrightarrow (Bx \wedge \forall y (Txy  \rightarrow (D_lx \leftrightarrow D_ly \wedge less(y,x) \rightarrow \neg By))\big), \\
& \forall x \big(B^{\sss 0}x \leftrightarrow (\neg Bx \wedge \forall y (Txy  \rightarrow (D_lx \leftrightarrow D_ly \wedge less(y,x) \rightarrow  By))\big), 
\end{eqnarray*}
We further use $B^{\sss c}$ to mark positions with local coordinates greater than the local coordinate of the elements marked $B^{\sss 0}$:
\begin{eqnarray*}
&\forall xy (Txy \rightarrow( D_lx \wedge D_ly \wedge B^{\sss 0}y \wedge less(y,x) \rightarrow B^{\sss c}x)).
\end{eqnarray*}

Finally we say that for two consecutive cliques, positions marked, respectively, with $B^{\sss 0}$ and $B^{\sss 1}$ have 
the same local coordinates,  and that the values of $B$ from positions marked with $B^{\sss c}$ in the first clique are copied to
the corresponding positions of the second clique. We here show formulas for the case when the first clique is in $D_0$:
\begin{eqnarray*}
&\forall xy (Txy \rightarrow (D_0x \wedge D_1y \wedge B^{\sss 0}x \wedge equal (x,y) \rightarrow B^{\sss 1}y)), \\
& \forall xy (Txy \rightarrow (D_0x \wedge D_1y \wedge equal(x,y) \wedge B^{\sss c}x \rightarrow (Bx \leftrightarrow By))
\end{eqnarray*}
Analogous formulas have to be written for the cases when the first clique is in $D_1, D_2$ and $D_3$, remembering that for the cases of $D_1$ and $D_3$
the guard 
$Tyx$ instead of $Txy$ should be used.

It is readily verified that the structure illustrated in Figure~\ref{f:altchain} is indeed a model of the outlined formula $\lambda_n^*$, and
moreover it must be a substructure  
of any model of $\lambda_n^*$. 

What is interesting, $\lambda_n^*$ has infinite models with only singleton cliques. To force infinite models of our formula
to contain cliques of at least exponential size we can use equality, and say that there is no pair of nodes in the same $D_l$ having
the same local coordinate and joined by $T$.

Returning to finite structures we finally observe 
that each finite model of $\lambda^*_n$ has doubly-exponentially many $T$-cliques 
that are distinguished by the sets of atomic $1$-types realized by its elements (for a formal definition of a $1$-type
see next section).
This contrasts with the case of equivalence relations: soon we will see that for any finitely satisfiable
\GFtEG{} formula one can find a finite model 
for which the collection containing for every equivalence
class the set of $1$-types realized in this class is bounded exponentially.
This is possible  even though, as we saw in Example \ref{klasa}, we must take into consideration classes of at least doubly exponential size. 
\end{example}

\subsection{Normal forms} \label{sec:normalforms}
In this paper, as in many other studies concerning two-variable logics, it is useful to consider `Scott-type' normal forms, allowing us to restrict the nesting of quantifiers to depth two. 
We will adapt the normal form introduced for \GFtTG{} in \citeN{ST04}.  There, a $\GFtTG$-formula is in normal form 
 if it is a conjunction of formulas of the following
	form: 	
	\begin{align*}
	\ce \qquad               &  \exists x (\ung(x) \wedge \psi(x))\\
	\cff    \qquad              & \forall x \forall y (\bing(x,y) \rightarrow \psi(x,y))\\
	\cfe \qquad  &\forall x (\ung(x) \rightarrow \exists y (\bing(x,y) \wedge \psi(x,y)))
	\end{align*}
	where  $\ung$ and $\bing$ are guards  and  $\psi$ is a quantifier-free  formula
	containing no special
	relation symbols. 
	We adopt the following convention for writing guard formulas:  $\ung(x)$ denotes a guard with exactly one variable $x$, $\bing(x,y)$ denotes a guard with both variables  $x,y$. Otherwise, we assume that $x$ and $y$ are allowed (but not required) to appear in $\psi(x,y)$.

Below we recall Lemma~2 from \citeN{ST04} that justifies the above normal form. We remark that the proof applies without changes to cases where the special relations symbols are required to satisfy other axioms, as considered in this paper. 

\begin{lemma}\label{lem:normalform}
	Let $\varphi$ be a \GFt{}-sentence with special guards over a signature $\sigma$. We can compute, in exponential time, a disjunction $\Psi=\bigvee_{i\in I} \psi_i$ of normal form sentences over a signature $\sigma'$ extending $\sigma$ by some additional unary symbols such that:
	(i) $\models \Psi \rightarrow \phi$; (ii) every model of $\phi$ can be expanded to a model of $\Psi$ by appropriately interpreting the additional unary symbols; (iii) $\sizeOf{\psi_i}=O(\sizeOf{\varphi} \log \sizeOf{\varphi})$ \textup{(}$i\in I$\textup{)}. 
\end{lemma}

Since our upper bounds are at least \NExpTime{} this lemma allows us to concentrate on normal form formulas.
Indeed, when necessary, we can compute $\Psi$ and non-deterministically guess one of its disjuncts. Also since $\phi$ and $\Psi$ are satisfiable over the same domains,  the same remark applies to bounds on the size of minimal finite models, since in this paper these bounds are at least exponential in the length of the formula.

Below we adapt the above normal form to fit our future purposes.
\begin{definition}\label{def:normal} 
We say that a \GFt-formula $\phi$ with  special guards over a signature $\sigma$ with special relation symbols $S_1, S_2, \ldots S_k $  is {\em in
normal form} if it is a conjunction of formulas of the following
form: 
\begin{align*}
\ce \qquad               &  \exists x (\ung(x) \wedge \psi(x))\\
\cffs \qquad      &\forall x \forall y (\eta(x,y) \rightarrow \psi(x,y))\\
\cfs  \qquad                          & \forall x (S_ixx \rightarrow \psi(x))\\
\cff    \qquad              & \forall x \forall y (\bing(x,y) \rightarrow \psi(x,y))\\
\cf     \qquad             & \forall x (\ung(x) \rightarrow \psi(x)) \\
\cfes \qquad  &\forall x (\ung(x) \rightarrow \exists y (\eta(x,y) \wedge \psi(x,y)))\\
\cfe \qquad  &\forall x (\ung(x) \rightarrow \exists y (\bing(x,y) \wedge \psi(x,y)))
\end{align*}
where  $\ung(x)$ and $\bing(x,y)$ are guards not using special symbols, $\eta(x,y)$ is one of 
the conjunctions $S_ixy\wedge S_iyx$, $S_ixy\wedge \neg S_iyx$ or $S_iyx\wedge \neg S_ixy$, and  
$\psi$ is a quantifier-free formula not containing special relation symbols. 

\end{definition}
When equality is available in the signature, we assume that conjuncts of types $\cfs$ and $\cf$ in the above definition are rewritten as conjuncts of type $\cffs$ and $\cff$. The subscript in $\cfs, \cffs$ and $\cfes$ indicates the conjuncts that essentially use special relation symbols. 

One can argue that the normal form formula is not in \GFt{} with special guards, as the subformulas $\eta(x,y)$ in $\cfes$-conjuncts  admit special symbols in positions that are not strictly guards.  
However, having $\eta(x,y)$ in the above defined way allows us to present our constructions in a more structured and transparent way. Moreover, this is not a restriction: one can  imagine that 
in the process of transforming a given \GFt{}-formula $\phi$ with special guards to its normal form, when $Sxy$ is a guard of some quantifier $Q$ we first replace $Sxy$ in $\phi$ by an equivalent formula $Sxy\wedge  (Syx \vee \neg Syx)$, and appropriately rearrange the resulting  formula using propositional tautologies and distribution laws for quantifiers. 
This process gives a \GFt{}-formula where the special relation symbols are restricted to guard positions of the forms $Sxy\wedge Syx$, $Sxy \wedge \neg Syx$ or $Syx \wedge \neg Sxy$. Naturally, Lemma \ref{lem:normalform} remains true for our adaptation of the normal form from \citeN{ST04}.

We note that the upper bounds in the paper are obtained when we allow this slightly extended syntax but for the corresponding lower bounds atomic guards suffice. 
Hence, by courtesy in the sequel when we write \GFt{} with special guards we refer to the logic with the slightly extended syntax, and we refer to the conjunctions denoted by $\eta(x,y)$ as guards.

For a normal form formula $\phi$ as in Definition~\ref{def:normal} we group together various conjuncts of $\phi$ and write: 
\begin{eqnarray*}
	\phiuniv&:=& \bigwedge\{\psi \mid ~ \psi \mbox{ is a conjunct of $\phi$  of type
		$\cffs$, $\cfs$, $\cff$, or $\cf$}\}\\
	\phisyms{S_i} &:=& \phiuniv \wedge \bigwedge\{\psi \mid ~ \psi \mbox{ is a conjunct of type $\cfes$ where }  \eta(x,y) = S_ixy\wedge S_iyx\}\\ 
		\phisym&:=&  \bigwedge_i   \phisyms{S_i} \wedge \bigwedge\{\psi \mid ~ \psi \mbox{ is a conjunct of $\phi$ of type $\ce$}\}\\
		\phieq &:=& \phisym \wedge  \bigwedge\{\psi \mid ~ \psi \mbox{ is a  conjunct of $\phi$ of type $\cfe$}\}.\\
\end{eqnarray*}
A special guard $\eta(x,y)$ is called {\em symmetric} when $\eta(x,y)=S_ixy \wedge S_iyx$,  otherwise, 
when $\eta(x,y)=S_ixy \wedge \neg S_iyx$, or $S_iyx \wedge \neg S_ixy$, it is called \emph{non-symmetric}. 
Using the above notation we can write equivalently: 
\begin{eqnarray*}
\phi & := & \phieq \wedge \bigwedge\{\psi \mid ~ \psi \mbox{ is a  conjunct of $\phi$ of type $\cfes$ with a non-symmetric guard} \}
\end{eqnarray*}
Note also that if $\phi$ is a \GFtEG{} normal form formula then  $\phi=\phieq$.

 To get some intuition behind $\phisyms{S_i}$ consider first the case where $S_i$ is an equivalence relation. In this case our intention is to gather 
in $\phisyms{S_i}$ those conjuncts of $\phi$
which describe what must be satisfied in \emph{every} equivalence class of $S_i$. And indeed, conjuncts of type $\ce$ are excluded from $\phisyms{S_i}$ since they need to be satisfied only
in \emph{some} $S_i$-classes; conjuncts of type $\cfes$ with guards $S_j$, $i \not=j$, are excluded since as mentioned above
we will always be able to build ramified models, in which every pair of distinct elements is connected by at most one special relation;
finally, conjuncts of type $\forall \exists$ are excluded since for every element we will always be able to provide witnesses for such conjuncts outside its $S_i$-equivalence class. If $S_i$ is not an equivalence but just an arbitrary transitive relation then $\phisyms{S_i}$ describes what must happen in every $S_i$-clique.  

Our notation also outlines a general scheme of model constructions in this paper: in order to build a model for $\phi$ we start by building a model for $\phisym$, then use it to construct a model for $\phieq$, and then proceed to build the final model for $\phi$.

 In next sections we often denote the special relation symbols required to be interpreted as equivalence
 relations using the letter $E$ (possibly with some decorations), and those required to be interpreted as transitive relations using the letter $T$.

\section{Types and Counting types}\label{sec:types-and-counting} 

In this Section we introduce the main technical notions. We start with the standard notion of types.

An {\em atomic} $1$-{\em type} over a signature $\sigma$ is a
maximal consistent set of atomic or negated atomic formulas over
$\sigma$ in  variable $x$. An {\em atomic} $2$-type
is a maximal consistent set of atomic or negated atomic formulas in  variables $x,y$, containing $x \not=y$.  
We often identify a type with the conjunction of its elements. Each $l$-type is a complete description  of a unique  (up to isomorphism) $\sigma$-structure with $l$-elements. 
Let $\str{A}$ be a $\sigma$-structure with universe $A$, and let
$a,b\in A$, $a\neq b$. We denote by $\type{\str{A}}{a}$ the unique atomic
1-type \emph{realized} in $\str{A}$ by the element $a$, i.e., the $1$-type $\alpha(x)$ such that $\str{A} \models \alpha(a)$; similarly by
$\type{\str{A}}{a,b}$ we denote the unique atomic 2-type realized in $\str{A}$ by
the pair $(a,b)$, i.e., the $2$-type $\beta(x,y)$ such that $\str{A} \models \beta(a,b)$.       

Let $\str{A}$ be a $\sigma$-structure and let $P$ be a binary relation symbol in  $\sigma$. A $P$-{\em clique} in $\str{A}$ is a maximal and non-empty set $C\subseteq A$ such that for all $a,b \in C$ we have $\str{A} \models Pab \wedge Pba$. If $P^{\str{A}}$ is an equivalence relation in $\str{A}$ then the notion of a $P$-clique coincides with the notion of an equivalence class. 
If  $C$ is a $P$-clique in $\str{A}$ then  the substructure $\str{A}\restr C$ is called {\em $P$-class}. In particular, if $A$ is a $P$-clique in $\str{A}$ then the whole structure $\str{A}$ is a $P$-class.

For  a given \GFt{}-formula $\phi$ we denote by  $\sigma_\phi$ the signature consisting precisely of the relation symbols used in $\phi$, with $S_1,\ldots,S_k$ denoting the special symbols in $\sigma_\phi$. We denote by $|\phi|$ the length of $\phi$ measured in a standard way, as the number of symbols used to write $\phi$; obviously $|\phi|\geq |\sigma_\phi |\geq k$.  We denote by  $\AAA_\phi$  the set of all atomic
1-types, and by $\BBB_\phi$---the set of all atomic 2-types over $\sigma_\phi$. 
Observe that $|\AAA_\phi|$ and $|\BBB_\phi|$ are at most exponential in $|\phi|$.        

We say that a special structure $\str{A}$  is {\em ramified}, if for every distinct special relation symbols $S, S'$ there are no two distinct elements $a,b \in A$ such that $\str{A} \models (Sab \wedge S'ab ) \vee (Sab \wedge S'ba)$.  \citeN{ST04} show that every satisfiable \GFtTG{}-formula has a ramified model. 
The constructions of our paper will imply that also every finitely satisfiable   \GFt-formula with special guards has a ramified model; in particular distinct equivalence classes will have at most singular intersections.

Below we introduce the crucial notions for this paper. They will be used in contexts in which a normal form formula $\phi$ is fixed, the letter $S$ will always denote a special relation symbol from $\sigma_{\phi}$.

\begin{definition}\label{d:counting-type} 
	A \emph{counting type} (over $\sigma_\phi$) is a function $\theta: \AAA_\phi \rightarrow \N$ not everywhere zero. 
	If for $n>0$ a counting type has all its values in  $\{0,1, \ldots, n\}$ then it is called an $n$-\emph{counting type}.
	An atomic
	type $\alpha \in \AAA_\phi$ {\em appears} in a counting type $\theta$ if $\theta(\alpha)>0$.
	
	\begin{enumerate}[(i)]
		\item Let $\theta$ be a counting type and $n>0$. The {\em cutting of} $\theta$ to $n$ (or the $n$-\emph{cutting} of $\theta$) 
		is the 
		unique $n$-counting type $\theta'$ such that for every $\alpha\in\AAA_\phi$, $\theta'(\alpha)=\min(\theta(\alpha),n)$. 
		\item
		We say that a counting type $\theta'$ {\em extends} a counting type $\theta$ if for every $\alpha \in \AAA_\phi$, 
		$\theta'(\alpha) \ge \theta(\alpha)$, if $\theta(\alpha)>0$,  and $\theta'(\alpha)=\theta(\alpha)=0$, otherwise. We say that $\theta'$ {\em safely extends} $\theta$ if $\theta'$ extends $\theta$ and additionally
		$\theta'(\alpha)=\theta(\alpha)$ if $\theta(\alpha)=1$. 
		\item We say that a finite class $\str{C}$ {\em realizes (has)} 
		 a counting type $\theta$ if for every $\alpha \in \AAA_\phi$, $\theta(\alpha)$ is the number of realizations of $\alpha$ in $\str{C}$. Occasionally we use the definition for infinite structures where we allow  $\theta(\alpha)=\infty$.      
	\end{enumerate}
	Note that a given class $\str{C}$ realizes precisely one counting type, denoted $\ctype{\str{C}}$. 
\end{definition}

\begin{definition}\label{d:admissible}	
	Let $\phi$ be a normal form formula as in Definition \ref{def:normal}.  We say that  
	$\alpha\in \AAA_\phi$  ($\beta\in \BBB_\phi$) is {\em $\phi_\forall$-admissible}, 
	if it is realized in some special one-element (two-element) $\sigma_\phi$-structure satisfying $\phiuniv$. 
We say that a counting type $\theta$ is {\em $\phisyms{S}$-admissible}, 
	if there exists an $S$-class $\str{C}$ whose
	counting type is $\theta$ and $\str{C} \models
	\phisyms{S}$. 
	Such an $S$-class $\str{C}$ is sometimes also called \emph{$\phisyms{S}$-admissible}.       
\end{definition}

Note that if $\str{A} \models \phi$ then every 1-type and every 2-type realized in $\str{A}$ is $\phiuniv$-admissible; moreover, when $S$ is transitive in $\str{A}$ 
then every $S$-class $\str{C}$ in $\str{A}$  and its counting type are
$\phisyms{S}$-admissible. 

The following fact justifies the notions introduced above. 

\begin{fact}\label{fact:models-admissible} 
	Let $\phi$ be a normal form \GFtTRG{}-formula 
	and let $\str{A}$ be a special $\sigma_\phi$-structure. Then $\str{A} \models \phisym$ if and only if the following
		conditions hold:
	\begin{enumerate}[(i)]
		\item $\str{A} \models \bigwedge\{\gamma \mid \text{$\gamma$ is a conjunct  of } \phi \text{ of type $\ce$}\}$,
		\item every 1-type and every 2-type realized in $\str{A}$ is $\phi_\forall$-admissible,
		\item for each special symbol $S_i$ in $\sigma$, each $S_i$-class in $\str{A}$ is $\phisyms{S_i}$-admissible. 
	\end{enumerate}
\end{fact}
\begin{proof}
	Straightforward.       
\end{proof}
We remark that the above fact remains true also when we do not require the special symbols to be reflexive, but we will use the fact only under this assumption.

Our intended models constructed in this paper will be {\em ramified}; below we introduce the notions of an $S$-splice for a 2-type and for an $S$-class that give objects that can be realized in such models.       

Let $\beta\in \BBB_\phi$.  We say that $\beta$ is  {\em binary-free}, if it contains negations of all binary atoms with two variables $x$ and $y$. Let $P$ be a binary relation symbol (special or not). By $\tsplice{\beta}{P}$ we denote the {\em $P$-splice} of $\beta$, namely the  unique 2-type obtained from $\beta$ by replacing each atomic formula of the form $S_ixy$ and $S_iyx$, where $1\leq i\leq k$, $S_i$ is special and $S_i\neq P$,   by its negation. 

For an $S$-class $\str{C}$ we denote by $\Csplice{\str{C}}{S}$ the $S$-{\em splice} of $\str{C}$, i.e., the class obtained from $\str{C}$ by replacing every 2-type $\beta$ realized in $\str{C}$ by $\tsplice{\beta}{S}$. We remark that this operation does not  change $1$-types of elements; in particular we do not remove self-loops.

We note  that for formulas where special relation symbols are restricted to guards the above operations on 2-types  and $S$-classes do not change their admissibility as shown in the following fact.
 \begin{fact}\label{fact:admisible-type}
	If  $\beta\in \BBB_\phi$ is $\phi_\forall$-admissible then $\tsplice{\beta}{P}$ is $\phi_\forall$-admissible, where $P$ is any binary relation symbol in $\sigma_{\phi}$.
	     	If $\alpha,\alpha'\in \AAA_\phi$ are $\phi_\forall$-admissible, then the unique binary-free 2-type $\beta'$ containing $\alpha(x)$ and $\alpha'(y)$ is $\phi_\forall$-admissible.   
	Moreover, for every $\phisyms{S_i}$-admissible $S_i$-class $\str{C}$, the $S_i$-splice of  $\str{C}$ is also $\phisyms{S_i}$-admissible.
\end{fact}

Given a normal form formula $\phi$ we set $M_\phi=3 |\AAA_\phi||\phi|^3$. 
Note that the value of $M_\phi$ is at most exponential in the size of $\phi$.
We establish some crucial properties of counting
types in the following lemma: parts (i) and (ii) say that admissibility transfers from counting types to their (safe) extensions; part (iii) says that 
admissibility of counting types with potentially large values transfers to their $M_\phi$-cuttings. 

\begin{lemma} \label{l:admissibility}
	Let $\phi$ be a normal form \GFt{}-formula with special guards, let $S$ be a special relation symbol in $\sigma_\phi$, and let $\theta$ be a 
	$\phisyms{S}$-admissible counting type.
	\begin{enumerate}[(i)]
		\item If $\phi$ does not use equality and $\theta'$ is a counting type
		extending $\theta$ then $\theta'$ is $\phisyms{S}$-admissible. 
		
		\item If  $\theta'$ is a counting type safely extending $\theta$ then $\theta'$ is $\phisyms{S}$-admissible. 
		\item If $\theta'$ is the cutting of $\theta$ to $M_\phi$ then  $\theta'$ is $\phisyms{S}$-admissible. 
	\end{enumerate}
\end{lemma}

\begin{proof}
	Let $\str{C}$ be an $S$-class of counting type $\theta$ such that $\str{C} \models \phisyms{S}$. 
		The 
	formula $\chi=\phisyms{S} \wedge \forall xy S xy$ can be treated as a
	normal form $\FOt$ formula and we can use techniques known from $\FOt{}$. In particular, to prove (i) and (ii) it suffices to notice that we can extend $\str{C}$ to a new model $\str{C'}$ of $\chi$ by adding a realization of any $\alpha\in\AAA_\phi$ such that, in case (i), $\theta(\alpha)\geq 1$, and in case (ii), $\theta(\alpha)\geq 2$. 
	Part (iii) can be seen as a reformulation of the small model property for \FOt{}.
	We briefly recall the crucial details below. 
	
	Case (i). Let $\theta(\alpha)\geq 1$.  Choose any $a\in C$ realizing $\alpha$.  Extend $\str{C}$ by a new element $a'$ defining $\type{\str{C}'}{a'}=\alpha$ and for every $b\in A$, $b\neq a$ define $\type{\str{C}'}{a',b}=\type{\str{C}}{a,b}$. 
	Complete $\str{C'}$ setting, for every binary relation symbol $Q\in \sigma_\phi$,  $\str{C}'\models Qaa'$ iff ${\str{C}}\models Qaa$ and $\str{C}'\models Qa'a$ iff ${\str{C}}\models Qaa$. 
	
	Case (ii). Let $\theta(\alpha)\geq 2$.  As before, choose any $a\in C$ realizing $\alpha$.  Extend $\str{C}$ by a new element $a'$ defining $\type{\str{C}'}{a'}=\alpha$ and for every $b\in A$, $b\neq a$, define $\type{\str{C}'}{a',b}=\type{\str{C}}{a,b}$.  To complete $\str{C'}$, find $c\in C$ such that $a\neq c$ and $\type{\str{C}}{c}=\alpha$, and define $\type{\str{C'}}{a,a'}=\type{\str{C}}{a,c}$. 
	
	It is routine to check that in both cases $\str{C'}\models \chi$.

	By adding a sufficient number of realizations of required 1-types, we obtain a $\phisyms{S}$-admissible $S$-class  that has counting type $\theta'$.

	To prove (iii) let $\theta'$ be the cutting of $\theta=\ctype{\str{C}}$  to $M_\phi$ and let $h$ be the number of conjuncts of the form $\cfes$ in $\phisyms{S}$
	(note that in such conjuncts we have $\eta(x,y) = Sxy \wedge Syx$). W.l.o.g.~assume $h \ge 2$.     
	
	We proceed by first selecting a sub-structure $\str{C}_0\subseteq \str{C}$ of bounded size having a counting type $\theta_0$ such that $\theta'$ safely extends $\theta_0$, and then modifying $\str{C}_0$ to obtain a structure $\str{C}_1\models \chi$ of the same counting type $\theta_0$. Admissibility of $\theta'$ will then follow from part (ii) of our lemma. 
	
	For each 1-type $\alpha$ that appears in $\theta$  mark $h$ distinct realisations $c_{\alpha,1}, \ldots, c_{\alpha,h}$      
	of $\alpha$ in $\str{C}$ (or all such realisations if $\alpha$ is realised less than $h$ times). Let $B_0$ be the set of all elements marked in this step. We mark additional elements to form $\str{C}_0$.
	
	For each $a \in B_0$ and each conjunct $\gamma$ of the form $\cfes$ in $\phisyms{S}$, 
	$\gamma=\forall x (\ung(x) \rightarrow \exists y (Sxy \wedge Syx \wedge \psi(x,y)))$,  if
	$\str{C} \models \ung(a) \wedge \neg \psi(a,a)$, then find a witness $b \in C$ such that $\str{C} \models \psi(a,b)$,
	and mark the element $b$.  Let $B_1$ be the set of all elements marked in this step. Thus, witnesses have now been found for all elements of $B_0$.
	
	Similarly, for all conjuncts of the form $\cfes$ in $\phisyms{S}$ we find witnesses for all elements from $B_1$ and mark them forming a set $B_2$. 
	
	Let $\str{C}_0=\str{C}\restr (B_0\cup B_1 \cup B_2)$ and let $\theta_0=\ctype{\str{C}_0}$. Obviously, $\theta$ safely extends $\theta_0$ and, 
	since $|B_0\cup B_1 \cup B_2|\leq |\AAA_\phi|(h+h^2+h^3|)\leq M_\phi$, also $\theta'$ safely extends $\theta_0$.  
		
	We now modify $\str{C}_0$ to yield a structure $\str{C}_1$ with the same domain, where only some 2-types between elements from $B_2$ and $B_0$ are redefined in order to ensure that all elements from $B_2$
		find witnesses 	for all conjuncts of the form $\cfes$ in $B_0$; in particular the 1-types of all elements are preserved, hence, also the counting type.
	
	 More precisely, consider any conjunct $\gamma_i$ ($1\leq i \leq h$) of the form $\cfes$, 
	$\gamma_i=\forall x (\ung(x) \rightarrow \exists y (Sxy \wedge Syx \wedge \psi(x,y)))$, and any element $a\in B_2$ such that
	$\str{C}_0 \models \ung(a) \wedge \neg \exists y \psi(a,y)$. Since $\str{C}\models \gamma_i$  there is  $b_i \in C$ such that $\str{C} \models \psi(a,b)$;  let $\alpha
	=\type{\str{C}}{b}$ and $\beta(x,y) =\type{\str{C}}{a,b}$. Our choice of $B_0$ ensures that $c_{\alpha,1}, \ldots, c_{\alpha,h}\in B_0$ are realizations of $\alpha$ in $\str{C}_0$; otherwise, there are less than $h$ realizations of $\alpha$ in $\str{C}$, and 	$\str{C}_0 \models \ung(a) \wedge \exists y \psi(a,y)$.
 We can, therefore,  replace in $\str{C}_0$ the 2-type $\type{\str{C}_0}{a,c_{\alpha,i}}$ by $\beta(x,y)$, providing a witness of $\gamma_i$ for $a$ in $B_0$ without violating any conjunct from $\phiuniv$, as $\beta$ is $\phiuniv$-admissible, and without disturbing other witnesses defined for $a$ for conjuncts $\gamma_j$ with $j\neq i$.   
  Moreover, by the choice of $B_1$ the element $c_{\alpha,i}$ does not require  any witnesses in $B_2$, hence the replacement of $\type{\str{C}_0}{a,c_{\alpha,i}}$ by $\beta$ does not destroy any essential witness for $c_{\alpha,i}$.

	 Denote the structure obtained by the above described modification by $\str{C}_1$. 
		From the construction it is clear that $\str{C}_1$ is an $S$-class,   $\ctype{\str{C}_1}=\theta_0$, and $\str{C}_1$ satisfies all conjuncts of the form $\phisyms{S}$, hence,
	$\str{C}_1$ is $\phisyms{S}$-admissible. 
	Since 
	$\theta'$ safely extends $\theta_0$,  applying part (ii) of 
	our lemma to $\theta'$ and $\theta_0$ the claim follows. 
\end{proof}

Part (iii) of Lemma \ref{l:admissibility} implies in particular the following.
\begin{remark}\label{cor:admissible}
	Let $\phi$ be a \GFt-formula with special guards in normal form and $S$ be a special symbol in $\sigma_\phi$. If $\str{A} \models \phi$ then 
	 for every $S$-class $\str{C}$ in $\str{A}$ the $M_\phi$-cutting of $\ctype{\str{C}}$ is $\phisyms{S}$-admissible.
	  \end{remark}

\section{Equivalence guards. The case without equality}\label{sec:eq-no-equality}

We now show that for \GFtEG{} the presence of the equality symbol is crucial both for
enforcing infinite models (cf.~Example \ref{ex:infiniteEG}) and enforcing finite models of doubly
exponential size (cf.~Example \ref{klasa}). This section can also be treated as a warm-up before more involved Section \ref{sec:eq-with-equality}. We show:

\begin{theorem}\label{theorem:noequality}
Every satisfiable \GFtEG{} formula without equality has a finite model
of exponentially bounded size.
\end{theorem}

\begin{proof}
Let $\phi$ be a satisfiable normal form \GFtEG{}-formula without equality, $E_1, \ldots, E_k$ be  the equivalence relation symbols in $\sigma_\phi$, 
and let $\str{B}$ be a (not necessarily finite) model of $\phi$.  
Recall that in $\phi$ the      
special guards $\eta(x,y)$ in conjuncts of the form $\cfes$ are symmetric and  can be simplified to $\eta(x,y)=E_ixy$ for some $i$. 

We show how to construct a finite,  exponentially bounded model $\str{A} \models \phi$.
First, we show how to build a model $\str{A}_0 \models \phisym$, i.e., a structure that can be partitioned into $\phisyms{S_i}$-admissible equivalence classes (for $i=1, \ldots, k$). 
Then, a model $\str{A} \models \phi$, will be obtained by taking
an appropriate number of copies of $\str{A}_0$ and setting the
2-types joining these copies in a simple way.

Let $\AAA_0$ be the set of $1$-types realized in $\str{B}$. For each $i$ ($1 \le i \le k$) choose a minimal collection $\str{C}^i_1, \ldots, \str{C}^i_{l_i}$ of
$E_i$-classes of $\str{A}$ such that for each $\alpha \in \AAA_0$ there is $j$ such that $\alpha$ is realized in $\str{C}_j^i$. 
For each class $\str{C}^i_j$ from this collection let $\theta^i_j$ be the $M_\phi$-cutting of $\ctype{\str{C}^i_j}$;
by Remark~\ref{cor:admissible}  
each $\theta^i_j$ is $\phisyms{E_i}$-admissible. 

We now define a \emph{base domain}, which is a non-empty  set of elements with their $1$-types predefined.
To built  $\str{A}_0$ we will then use some number of copies of the base domain.
For every $\alpha \in \AAA_0$ let $n_\alpha= \max_i \sum_{j=1}^{l_i} \theta^i_j(\alpha)$. 
We put $n_\alpha$ copies of realizations of each $\alpha$ into the base domain.
 Let $a_0, \ldots, a_{l-1}$ be an enumeration 
of the elements of the so obtained base domain.  Note that $l$ is bounded by $M_\phi |\AAA_0|^2 \le M_\phi |\AAA_\phi|^2$, exponentially in $|\phi|$.

Let us now observe that for each $i$ ($1 \le i \le k$) the base domain can be partitioned into subsets on which  ramified $\phisyms{E_i}$-admissible 
$E_i$-classes can be built. 
Fix $1 \le i \le k$.
By the choice of the numbers $n_\alpha$ we have enough elements of every $1$-type in the base domain to construct  sets $P^i_1, \ldots, P^{i}_{l_i}$, such that $P^i_j$
contains precisely $\theta^i_j(\alpha)$ realizations of $\alpha$, for every $\alpha\in \AAA_0$.  After this step some elements of the base domain (call them \emph{redundant})
may  be not     
assigned to any $P_j^i$.  For each redundant element $a$, assuming its predefined $1$-type is $\alpha$, choose a set $P^i_j$ containing a realization of $\alpha$ and join $a$ to $P^i_j$.
After this step, the sets $P^i_j$ form a partition of the base domain. 
Let ${\overline\theta}^i_j$ be the counting type such that ${\overline\theta}^i_j(\alpha)$ is   the number of elements of 1-type $\alpha$ in $P^i_j$. Obviously, ${\overline\theta}^i_j$ extends ${\theta}^i_j$, so by  part (i) of  Lemma \ref{l:admissibility} ${\overline\theta}^i_j$ is $\phisyms{E_i}$-admissible, and we can expand $P^i_j$ to a $\phisyms{E_i}$-admissible class  (c.f., Definition~\ref{d:admissible}). Finally, due to Fact \ref{fact:admisible-type}
each such class can be replaced by its $E_i$-splice, without affecting its $\phisyms{E_i}$-admissibility.

We emphasise that it is the step of adjoining redundant elements to the sets $P^i_j$  in which the absence of equality is crucial. Indeed, with equality we can forbid more than one
realization of a $1$-type in a class, e.g., by saying $\forall xy (E_ixy \rightarrow (Px \wedge Py \rightarrow x=y))$. Thus, for the case with equality (Section 
\ref{sec:eq-with-equality}) we will need a more careful strategy of building finite models. In particular, instead of part (i) we will use a less convenient part (ii) of Lemma \ref{l:admissibility}.

 To form $\str{A}_0$ we put
copies of elements from the base domain into the nodes of a
$k$-dimensional grid:
at location $(x_1, \ldots, x_k)$, for $0 \le x_i < l$, we
put an element whose 1-type equals the 1-type of $a_s$,
where $s=(x_1 + \ldots + x_k) \; \mbox{mod}\; l$.
This way each line of the grid (set of tuples having the same coordinates on all but one position) 
consists precisely of a copy of the base domain. 
Figure~\ref{grid} presents the arrangement of elements for $l=5$ and $k=2$.

\begin{figure}[htb]
\begin{center}
\begin{tikzpicture}[scale=1.2]

\begin{scope}[transparency group]
\begin{scope}[blend mode=multiply]

%poziome pelne
\foreach \x/\y/\l in {1/1/3, 3/4/3, 2/5/3, 4/1/2, 3/2/2, 2/3/2, 1/4/2}{
\draw[draw=none, rounded corners=5pt, fill=black!70!white, opacity=0.5] (\x-0.15, \y-0.15) rectangle (\x+\l+0.15-1, \y+0.15) {};
}

%poziome bez lewego łuku
\foreach \x/\y/\l in {1/2/2, 1/3/1, 1/5/1}{
\draw[draw=none, fill=black!70!white, opacity=0.5]
  (\x-0.25, \y-0.15) {[rounded corners=5pt] --
  ++(\l+0.4-1,0)  -- 
  ++(0,0.3)} --
  ++(-\l-0.5+1, 0) --
  cycle;
	}

%poziome bez prawego luku
\foreach \x/\y/\l in {5/2/1, 4/3/2, 5/5/1}{
\draw[draw=none, fill=black!70!white, opacity=0.5]
  (\x+\l-1+0.25, \y+0.15) {[rounded corners=5pt] --
  ++(-\l-0.4+1,0)  -- 
  ++(0,-0.3)} --
  ++(\l+0.5-1, 0) --
  cycle;
	}
	
	%pionowe pelne
	\foreach \x/\y/\l in {1/1/2, 1/3/2, 2/2/2, 3/1/2, 3/4/2, 4/3/2, 5/2/2, 5/4/2, 1/5/1, 2/4/1, 3/3/1, 4/2/1, 5/1/1} {
\draw[draw=none, rounded corners=5pt, fill=black!30!white, opacity=0.5] (\x-0.15, \y-0.15) rectangle (\x+0.15, \y+\l-1+0.15) {};
}

%pionowe bez gornego luku
\foreach \x/\y/\l in {2/5/1, 4/5/1}{
\draw[draw=none, fill=black!30!white, opacity=0.5]
  (\x-0.15, \y+0.35) {[rounded corners=5pt] --
  ++(0, -0.5)  -- 
  ++(0.3,0)} --
  ++(0, 0.4) --
  cycle;
	}

%pionowe bez dolnego luku
\foreach \x/\y/\l in {2/1/1, 4/1/1}{
\draw[draw=none, fill=black!20!white, opacity=0.5]
  (\x+0.15, \y+0.65-1) {[rounded corners=5pt] --
  ++(0, 0.5)  -- 
  ++(-0.3,0)} --
  ++(0, -0.4) --
  cycle;
	}
\end{scope}
\end{scope}

\foreach \x in {1,2,3,4,5}
   \foreach \y in {1,2,3,4,5}{
      \filldraw[fill=black] (\x, \y) circle (0.05);  
			\pgfmathtruncatemacro \aux {\intcalcMod{\x+\y-2}{5}};
			\coordinate [label=center:$a_{\aux}$] (A) at ($(\x-0.3,\y+0.3)$); 			
 }

\end{tikzpicture}
\caption{Arrangement of elements on the two-dimensional grid and an illustrative division into classes. $E_1$-classes are dark grey,
$E_2$-classes are light grey.} 
\label{grid}
\end{center}
\end{figure}
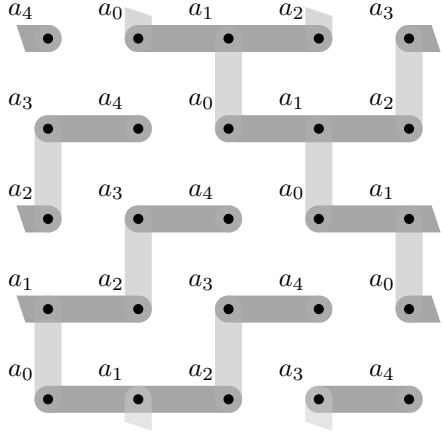

 Now, horizontal lines of the grid are partitioned to form ramified $\phisyms{E_1}$-admissible  $E_1$-classes, 
which is possible, as explained above.
Similarly, vertical lines are partitioned to form ramified $\phisyms{E_2}$-admissible $E_2$-classes, lines going in the third dimension -- to form ramified $\phisyms{E_3}$-admissible $E_3$-classes, and so on (cf.~Figure~\ref{grid}). We
complete the definition of $\str{A}_0$ by setting all 2-types for
pairs of distinct elements not belonging to the same class as binary-free types. These types are $\phi_\forall$-admissible and do not affect the equivalence classes defined.

It is probably worth remarking why we use a grid instead of a single copy of the base domain. An attempt of forming a model just of a single copy of the base domain
would require joining some pairs of distinct elements by more than one equivalence relation, i.e., to form a model which is not ramified. This however may be forbidden, e.g., in the presence of two equivalences, by conjuncts
$\forall xy (E_1xy \rightarrow x \not=y \rightarrow Rxy) \wedge \forall xy (E_2xy \rightarrow x \not=y \rightarrow \neg Rxy)$, where $R$ is a non-special binary symbol.
    On the other hand, it follows from our construction that every finitely satisfiable sentence has a ramified model.

Observe that the structure $\str{A}_0$ satisfies conditions (i)-(iii) of Fact \ref{fact:models-admissible}. Indeed, conditions (ii) and (iii) hold, as explained, and condition (i) holds as every $\alpha\in \AAA_0$ is realized in  $\str{A}_0$ (by the choice of the collection of classes at the beginning of the construction).  Hence, $\str{A}_0 \models \phisym$. 

To construct $\str{A}$ we take three sets of copies of $\str{A}_0$,
each consisting of $h$ elements, where $h$ is the number of
conjuncts of $\phi$ of the form $\cfe$ (and thus the maximal number of witnesses an element may require for $\cfe$-conjuncts);
   w.l.o.g., we can suppose $h>0$. 
This
 step is reminiscent of the construction
of the small model for \FOt{} from \citeN{GKV97}. In this step we provide witnesses for all conjuncts of type $\cfe$ of $\phi$ without
 changing the $E_i$-classes defined in the copies of $\str{A}_0$; this will ensure that $\str{A}\models \phi$. 

Let $A=A_0 \times \{0,1,\ldots,h-1\} \times \{0,1,2\}$. We set
$\type{\str{A}}{a,i,j}:=\type{\str{A}_0}{a}$  and
$\type{\str{A}}{(a,i,j),(b,i,j)}:=\type{\str{A}_0}{a,b}$, for all $i,j$ ($0\leq i<h$, $0\leq j<3$). 
Witnesses for all elements in $A$ for conjuncts of $\phi$ of the
form $\cfe$ are provided in a circular way: elements from $A_0
\times \{0,1,\ldots,h-1\} \times \{j\}$  find witnesses in $A_0
\times \{0,1,\ldots,h-1\} \times \{j'\}$, where $j' = j + 1 \;
\mbox{mod} \; 3$. Consider an element $(a,i,j)$.
Let $\delta_0, \ldots, \delta_{h-1}$ be all conjuncts of $\phi$ of
the form $\cfe$,  $\delta_m=\forall x (\ung(x) \rightarrow \exists y
(\bing(x,y) \wedge \psi(x,y)))$. 
Find in $\str{A}$ an element $b$ such that $\type{\str{A}}{b}=\type{\str{A}_0}{a}$.
For $0 \le m < h$, if $\str{A} \models \ung(b) \wedge \neg(\bing(b,b) \wedge \psi(b,b))$,  choose a witness $b_m \in A$ 
such that $\str{A} \models \ung(b) \rightarrow 
(\bing(b,b_m) \wedge \psi(b,b_m))$. Let $\alpha_m=\type{\str{A}}{b_m}$ 
and let $\beta_m=\tsplice{\type{\str{A}}{b,b_m}}{Q}$, where $Q$ is the relation symbol of the atom $\bing(x,y)$.  Let $b_m$ 
be an element of type $\alpha_m$ in $\str{A}_0$.
We set $\type{\str{A}}{(a,i,j), (b_m,m,j+1 \;
\mbox{mod} \; 3 ) }:=\beta_m$. 

 This circular scheme guarantees no
conflict in setting 2-types, since a 2-type for a pair of elements
in $\str{A}$ is defined at most once. Moreover, in this step we use only splices of 2-types from the original model  that do not affect the equivalence classes and by Fact \ref{fact:admisible-type} are $\phi_\forall$-admissible.  We complete the structure
by setting 2-types for all remaining pairs of elements. Let
$(a,i,j)$, $(a',i',j')$ be such a pair. Let $\alpha=\type{\str{A}_0}{a}$,
$\alpha'=\type{\str{A}_0}{a'}$. Let $\beta$ be the unique binary-free 2-type containing $\alpha(x)$ and
$\alpha(y)$. We set
$\type{\str{A}}{(a,i,j),(a',i',j')}:=\beta$. 

This finishes our construction of $\str{A}$. 
To see that $\str{A}$ is indeed a model of $\phi$ note that conjuncts of the form $\ce$ are satisfied since $\str{A}$ realizes
all the $1$-types realized in the original model $\str{B}$,  conjuncts of the form $\cfes$ are satisfied 
since all $E_i$-classes ($i=1, \ldots, k$) in $\str{A}$ are $\phisyms{E_i}$-admissible, the conjuncts of the form $\cfe$ are taken care  about in the  two previous paragraphs,      
and conjuncts of the form $\cff$ are satisfied since all $2$-types in $\str{A}$  are either explicitly required to be $\phi_\forall$-admissible 
or are binary-free.
 Note that $|A_0| \le (M_\phi|\AAA_0|^2)^k$ and
thus the size of $\str{A}$ is bounded by $3h (M_\phi|\AAA_\phi|^2)^k$, exponentially in $|\phi|$. 
\end{proof} 

Having Theorem \ref{theorem:noequality} proved 
we can state the following.
\begin{corollary}\label{c:nexpnoeq}
The (finite) satisfiability problem for \GFtEG{} without equality is \NExpTime-complete. 
\end{corollary}
\begin{proof}
The lower bound follows from the lower bound for the 
	fragment with one equivalence relation  (that enjoys the finite
	model property) which holds without equality \cite{Kie05}. The procedure justifying the upper bound
	takes a formula $\phi$, converts it to normal form $\phi'$, guesses an exponentially bounded structure
	and verifies that it is a model of $\phi'$.
\end{proof}

\section{Equivalence Guards. The case with equality}\label{sec:eq-with-equality}

In this section we show that every finitely satisfiable
\GFtEG-formula has a model of bounded (double exponential) size and that 
the finite satisfiability problem for \GFtEG{} is \NExpTime-complete. 
We remark here that the case of finite satisfiability is much more
complicated than that of general satisfiability. Indeed, every
satisfiable formula has a (usually infinite) tree-like model with equivalence classes
bounded exponentially. To check if a given normal form formula has a model it suffices to guess a set of
1-types which are going to be realized, verify that for a realization of each 1-type one can build
its (small) $E_i$-classes (independently for different $i$), and check some additional simple conditions. 
 \citeN{Kie05} discusses details of this construction. For finite satisfiability we need
to develop a different and more involved approach.
Recall in particular, that we have to take into consideration equivalence classes of doubly exponential size, as shown in Example \ref{klasa}.  Since we are aiming at an $\NExpTime$ upper bound, similarly as in the previous section 
we approximate counting types of classes by 
their $M_\phi$-cuttings  (where $M_\phi$ is the number defined in Section~\ref{sec:types-and-counting} 
that admits application of part (iii) in Lemma~\ref{l:admissibility}). However, the number of $M_\phi$-counting types is doubly exponential, so we cannot work on this level of abstraction directly and we reduce the finite satisfiability problem to linear/integer programming. The proof is presented in two parts. In Subsection~\ref{ss:sysstems-of-linear} we identify  conditions of tuples of sets of $M_\phi$-counting types that obviously hold when 
these sets are taken from models of \GFtEG{}-formulas.
These conditions are formulated in terms of linear inequalities. In Subsection~\ref{ss:neat-model} we show the core part of the reduction, namely that checking solutions of these inequalities is sufficient for finite satisfiability.  

The idea of reducing the (finite) satisfiability problem of a logic to linear (integer) programming is not new.  It was, e.g., employed by \citeN{Cal96}, \citeN{LST05} and Pratt-Hartmann  \citeyear{PH05,PH07}
to establish the complexity of some fragments of the  two-variable logic  with counting quantifiers. On the other hand the apparatus of counting types
and their $M_\phi$-cuttings was introduced in the conference version of this paper \cite{KT07} and it was later adopted (after some changes in terminology)
in \citeN{KMP-HT14} and \citeN{KP-HT15}
to deal with related logics with equivalence relations.

Let us now recall some notions and results concerning linear and integer programming.

 A {\em linear equation} ({\em inequality}) is an expression $t_1 = t_2$ ($t_1 \geq t_2$),
where $t_1$ and $t_2$ are linear terms with coefficients in $\N$.
We are going to work with systems containing both equations and inequalities;  for simplicity
we will call them just \emph{systems of inequalities}, as any equation can be presented as two inequalities in an obvious way.
Given a system $\Gamma$ of linear inequalities, we take the
{\em size} of $\Gamma$ 
to be the total number of
bits required to write $\Gamma$ in standard notation; notice that the size of $\Gamma$
may be much larger than the number of inequalities in $\Gamma$.  

The problem {\em linear } ({\em integer})
	{\em programming} is as follows: given a system $\Gamma$ of linear
	inequalities decide  if $\Gamma$ has a solution over
	$\Q$ (over $\N$). 
It is well-known that liner programming is in \PTime{} \cite{khachiyan79} and integer programming is in \NP{} \cite{BoroshandTreybig76}. 
We recall the following useful bound on solutions of systems of linear inequalities. 

\begin{proposition}[\cite{Papa81}]\label{pr:papadimitriou}
	Let $\Gamma$ be a system of $m$ linear inequalities in $n$ unknowns, and let the
	coefficients and constants that appear in the inequalities be in $\{-a,  \ldots, a-1, a\}$. If $\Gamma$
	admits a non-negative integer solution, then it also admits one in which the values assigned to the
	unknowns are all bounded by $n(m \cdot a)^{2m+1}$.
\end{proposition}

\subsection{Systems of linear inequalities}\label{ss:sysstems-of-linear}
 
We assume that $E_1, \ldots, E_k$ are equivalence relation symbols available in the signature $\sigma_\phi$ of a normal form \GFtEG-formula $\phi$. Recalling the notation after Definition~\ref{def:normal} we note that we have now $\phi=\phieq$. 

We work with tuples of sets of $n$-counting types (for some $n >0$) of the form 
$\bnTheta= (\nTheta^{E_1}, \nTheta^{E_2}, \ldots, \nTheta^{E_k})$.
We say that an atomic type $\alpha \in \AAA_\phi$ {\em appears} in $\nTheta^{E_i}$ if 
$\alpha$ appears in some $\theta \in \nTheta^{E_i}$,  
and $\alpha$
appears in $\bnTheta$ if $\alpha$ appears in $\nTheta^{E_i}$ for some $i$. 
We denote by $\AAA_{\bnTheta}$  the set of 1-types that appear in $\bnTheta$.

\begin{definition} \label{d:tupleadmissibility}
Let 
$\bnTheta= (\nTheta^{E_1}, \nTheta^{E_2}, \ldots, \nTheta^{E_k})$
be a tuple of  non-empty sets of $n$-counting types (for some $n>0$)
and $\phi$---a normal form \GFtEG{} formula. We say that
$\bnTheta$ is $\phi$-admissible 
     if the following conditions hold:
\begin{enumerate}[(i)]
\item each $\theta \in \nTheta^{E_i}$ is $\phisyms{E_i}$-admissible, ($i=1, \ldots, k$)
\item for every conjunct $\delta$ of $\phi$ of the form $\ce$, $\delta=\exists x (\ung(x) \wedge \psi(x))$
there is $\alpha \in {\AAA_{\bnTheta}}$ such that $\alpha \models \ung(x) \wedge \psi(x)$,
\item for  every conjunct $\delta$ of $\phi$ of the form  $\cfe$, $\delta=\forall x
(\ung(x) \rightarrow \exists y (\bing(x,y) \wedge \psi(x,y)))$, for every
$\alpha \in \AAA_{\bnTheta}$ such that $\alpha\models \ung(x)$ and $\alpha \not\models \bing(x,x) \wedge \psi(x,x)$  there is $\alpha' \in {\AAA_{\bnTheta}}$
such that for some $\phi_\forall$-admissible 2-type $\beta
\in \BBB$ such that $\alpha(x) \subseteq \beta$, $\alpha'(y) \subseteq
\beta$, we have  $\beta
\models  \bing(x,y) \wedge   \psi(x,y) \wedge \neg E_1xy \wedge \neg E_1yx \wedge \ldots \wedge \neg E_kxy \wedge \neg E_kyx$. 
\end{enumerate}
\end{definition}

Recall $M_\phi = 3|\AAA_\phi||\phi|^3$.
The following observation is straightforward (cf. Remark \ref{cor:admissible}).
 
\begin{claim}
Let  $\phi$ be a \GFtEG{}-formula $\phi$ in normal form, $\str{A}$ be a finite model of $\phi$ and
$\bnTheta = (\nTheta^{E_1},
\nTheta^{E_2}, \ldots, \nTheta^{E_k})$, where $\nTheta^{E_i}$ is the set of the $M_\phi$-cuttings of the counting types realized in  $\str{A}$ by $E_i$-classes. Then
$\bnTheta$ is $\phi$-admissible.
\end{claim}

Given a tuple of sets of $M_\phi$-counting types 
$\bnTheta= (\nTheta^{E_1}, \nTheta^{E_2}, \ldots, \nTheta^{E_k})$,
we say that $\alpha \in \AAA_\phi$ is {\em royal} for $E_i$-classes, if $\alpha$
appears in $\nTheta^{E_i}$ and, for every $\theta \in
\nTheta^{E_i}$, we have $\theta(\alpha) < M_\phi$. 
In our construction we will be adding  realizations of non-royal $1$-types to  some $\phisyms{E_i}$-admissible $E_i$-classes, 
without affecting their admissibility. This will be possible due to part (ii) of Lemma \ref{l:admissibility}. 
For readers familiar with the classical paper on \FOt{} by \cite{GKV97} it is worth commenting that our definition of royal types might seem excessive as the obvious candidates for royal types are those that do not occur more than once in $E_i$-classes. Our definition, however, allows us not to bother with some details, and is sufficient for the complexity bounds we want to show.

 We associate with $\bnTheta$ a system
 of linear inequalities $\Gamma_{\bnTheta}$ defined below  describing conditions that hold when $\bnTheta$ contains the $M_\phi$-cuttings of all counting types realized in some finite model of $\phi$. The solutions of the constructed system will suggest the number of elements (and their $1$-types)  in a \emph{base domain}. Models constructed in the next subsection  will be built using some number of copies of the base domain, as in the case without equality.
 
  For each $i$ and each
 $\theta \in \nTheta^{E_i}$ we use a variable $X^{E_i}_{\theta}$,
 whose purpose is to suggest    
the number of all $E_i$-classes $\str{C}$ such that the $M_\phi$-cutting of $\ctype{\str{C}}$ equals $\theta$, 
 and, for each $\alpha\in \AAA_{\bnTheta}$, we use a variable
 $Y^{E_i}_\alpha$ to count how many elements of $1$-type $\alpha$ are required to
form the suggested $E_i$-classes.

For every  $\alpha\in \AAA_{\bnTheta}$ and every $i$ ($1\leq i\leq k$), we put to
$\Gamma_{\bnTheta}$ the following equations and inequalities:
\begin{itemize}
	
\item $Y^{E_i}_\alpha$ is the  \emph{suggested}  number of realizations of $\alpha$
     in $E_i$-classes:
\begin{equation*}
{ Y^{E_i}_\alpha=\sum_{\theta \in \nTheta^{E_i}} \theta(\alpha) \cdot
X^{E_i}_{\theta}}, \tag{E0}
\end{equation*}

\item

if $\alpha$ is royal for $E_i$-classes, then there is an $E_i$-class in which $\alpha$ appears: \quad 
\begin{equation*}
 Y^{E_i}_\alpha \ge 1,  \tag{E1a}
\end{equation*}

\item
if $\alpha$ is not royal for $E_i$-classes, then there is an $E_i$-class in which $\alpha$ appears at least $M_\phi$-times: 
\begin{equation*}
\sum_{\theta \in \nTheta^{E_i}, \theta(\alpha)=M_\phi} X_{\theta}^{E_i} \ge
1,  \tag{E1b}
\end{equation*}

\item  if $\alpha$ is royal for $E_i$-classes, then, for every $j \not=i$, the suggested number of realizations of $\alpha$ in  $E_i$-classes is at least equal
to the suggested number of realizations of $\alpha$ in $E_j$-classes:
\begin{equation*}
Y^{E_i}_\alpha \ge Y^{E_j}_\alpha \tag{E2}.
\end{equation*}
\end{itemize}

The intuition behind the inequalities (E2) is as follows: when a 1-type appears in $\nTheta^{E_i}$ and is royal for $E_i$-classes then $Y^{E_i}_\alpha$
 corresponds to the {\em exact} number of realisation of $\alpha$ in $E_i$-classes. 
 When $\alpha$ is not royal, then  $Y^{E_i}_\alpha$ is just a lower bound, as in this case inequality (E1b) will enforce our model to contain  an $E_i$-class $\str{C}$ 
with at least $M_\phi$-realizations of $\alpha$, and we will then be allowed to extend it by additional realizations of $\alpha$ (which will not affect the  $M_\phi$-cutting of $\ctype{\str{C}})$.
 Please note that 
 if $\alpha$ is royal for both $E_i$-classes
 and $E_j$-classes ($i \not= j$) then we write two inequalities of
 the form (E2) that give $Y^{E_i}_\alpha = Y^{E_j}_\alpha$.

One can easily check that the following  holds.
\begin{proposition}\label{f:model-to-solution}
Let  $\phi$ be a \GFtEG{}-formula $\phi$ in normal form, $\str{A}$ be a finite model of $\phi$ and
$\bnTheta = (\nTheta^{E_1},
\nTheta^{E_2}, \ldots, \nTheta^{E_k})$, where $\nTheta^{E_i}$ is the set of the $M_\phi$-cuttings of the counting types realized in  $\str{A}$ by $E_i$-classes. Then $\Gamma_{\bnTheta}$ has (in particular) the following non-negative integer solution:

$X^{E_i}_{\theta} =r^{E_i}_{\theta} = |\{\str{C}\mid \str{C} \text{ is an $E_i$-class in }  \str{A} \text{ and }  \theta \text{ is the $M_\phi$-cutting of } \ctype{\str{C}}\}|$,    

$Y^{E_i}_\alpha= \sum_{\theta \in \nTheta^{E_i}} \theta(\alpha)
\cdot r^{E_i}_{\theta}$.
\end{proposition}

Our aim now is to show that solutions of $\Gamma_{\bnTheta}$  correspond to certain {\em neat} models of $\phi$, defined below. 

\begin{definition}\label{def:neat}
	Let $\phi$	be a \GFtEG-sentence in normal form.
	A special ramified  structure $\fA\models \phi$ is a 
	{\em \neat model} for $\phi$ if the following conditions hold
	\begin{enumerate}[(i)]
		\item  the number of  distinct $M_\phi$-cuttings of counting types realized in $\fA$ is bounded by $\mathfrak{f}(|\phi|)$, for a fixed exponential function $\mathfrak{f}$. 
		\item   the size of $A$ is at most doubly exponential in $|\phi|$. 
	\end{enumerate}
\end{definition}
One could refer at this point to Example~\ref{ex:large-no-equality}, where it is shown that condition (i) of the above definition could not be obtained if the special symbols of $\phi$ were not declared to be equivalences, but rather arbitrary transitive relations.   

For technical reasons we additionally distinguish an equivalence relation $E_i$     
and a single $M_\phi$-counting type $\theta\in\nTheta^{E_i}$ that will be required to appear in a model (as the
$M_\phi$-cutting of the counting type of some $E_i$-class).     
This will become important in Section \ref{sec:transitive}.
Accordingly,  let $\uGamma{\bnTheta}{i,\theta}$  
be the system of linear inequalities obtained from $\Gamma_{\bnTheta}$ by adding the inequality: 
\begin{equation*}
X^{E_i}_{\theta} \geq 1 \tag{E3}.
\end{equation*}

Observe that the number of inequalities in $\Gamma_{\bnTheta}(i,\theta)$ is polynomial in  $|\AAA_\phi|$, 
and
the number of variables is  bounded by $|\AAA_\phi|+(M_\phi+1)^{|\AAA_\phi|}$ (i.e.,~doubly exponential in $|\phi|$). 
Let us also make a small observation concerning solutions of 
$\uGamma{\bnTheta}{i,\theta}$ in the lemma below. 

\begin{lemma}\label{lem:equations} The following conditions are equivalent 
	\begin{enumerate}[(i)]
	\item $\uGamma{\bnTheta}{i,\theta}$ has a non-negative integer solution.
	\item $\uGamma{\bnTheta}{i,\theta}$ has a non-negative
	rational solution. 
	\item $\uGamma{\bnTheta}{i,\theta}$ has 
	a non-negative integer solution with at most $m$ non-zero unknowns and in which all values are bounded by 
	$m(m\cdot M_\phi)^{2m+1}$, where  $m$ is the number of inequalities in $\uGamma{\bnTheta}{i,\theta}$. 
	\end{enumerate}
\end{lemma}
\begin{proof}
(i) $\Rightarrow$ (ii) and (iii) $\Rightarrow$ (i) are obvious.

(ii) $\Rightarrow$ (iii): It follows from basic facts of algebra  that if $\uGamma{\nTheta}{i,\theta}$ has a non-negative rational solution, then it has one in which the number of non-zero unknowns is not greater than the number of inequalities (see, e.g., \citeN{Paris}, Chapter 10).
Additionally, observe that all equations  and inequalities  in $\uGamma{\bnTheta}{i,\theta}$ are actually either  of the
form $\sum_i c_i x_i =0$ or $\sum_i c_i x_i\ge b$, with $b\ge 0$.
So, if $\uGamma{\nTheta}{i,\theta}$ has a rational solution, one can get an integer
solution by multiplying the rational solution  by the product of all the denominators. 
Let $X$ be such a solution. 
 Consider now the system  $\Gamma'$   obtained from $\uGamma{\bnTheta}{i,\theta}$ by inserting zeroes for all unknowns that have value zero in $X$. $\Gamma'$ has now at most $m$ unknowns. 
 Obviously, the above  solution of $\uGamma{\bnTheta}{i,\theta}$ restricted to the unknowns of $\Gamma'$ is a solution of $\Gamma'$. Applying Proposition \ref{pr:papadimitriou}, $\Gamma'$ has a solution $Y$ where all values are bounded as required. $Y$ can be extended to a solution of $\uGamma{\bnTheta}{i,\theta}$ by adding zero values for all unknowns that do not appear in $\Gamma'$. 
\end{proof}

\subsection{Neat model construction}\label{ss:neat-model}
Now we proceed with the main task.

\begin{lemma}\label{lem:from-solution-to-model} 
	Let  $\phi$ be a \GFtEG{}-formula $\phi$ in normal form. 
	 Let
	$\bnTheta = (\nTheta^{E_1},
	\nTheta^{E_2}, \ldots, \nTheta^{E_k})$ be a $\phi$-admissible  tuple of sets of $M_\phi$-counting types. 
	Let $\theta\in \nTheta^{E_i}$ for some $E_i$. 
	If the system $\uGamma{\bnTheta}{i,\theta}$ has a non-negative integer solution, then there is a neat model $\str{A}\models \phi$ such that
\begin{enumerate}[(i)]
\item 
for every $E_i$, for every $E_i$-class $\str{C}$ of $\str{A}$ the  $M_\phi$-cutting  of $\ctype{\str{C}}$ is in $\nTheta^{E_i}$,
\item  $\theta$ equals the $M_\phi$-cutting of $\ctype{\str{C}}$ for some $E_i$-class $\str{C}$ of $\str{A}$.

\end{enumerate}	
	\end{lemma}
	\begin{proof} Let 
		$X^{E_i}_{\theta} = r^{E_i}_{\theta}$ and
		$Y^{E_i}_\alpha=r^{E_i}_\alpha$,  ($1\leq i\leq k$, $\theta\in \bnTheta$) be a bounded  non-negative integer solution of  $\uGamma{\bnTheta}{i,\theta}$ as promised by part (iii) of  Lemma~\ref{lem:equations}.  We denote by $B$ the maximum value from the solution; it is ensured that $B$ is doubly exponential in $|\phi|$. 
		We will build a model $\str{A}$ for $\phi$ in which for each $i$ and $\theta\in \bnTheta^{E_i}$ there will be an $E_i$-class having a counting type which cut  to  $M_\phi$ equals      
		$\theta$, only if $r^{E_i}_{\theta}>0$.     
		This will ensure that the number of distinct $M_\phi$-cuttings of counting types       
		 realized in $\str{A}$ will be properly bounded, as required by part (i)  of Definition~\ref{def:neat}. 
	
		The construction of the required model $\str{A}\models \phi$ comprises of similar steps as the construction  in the proof of   Theorem \ref{theorem:noequality}. 
	 
	 We first build a model $\str{A}_0 \models \phisym$. Again as in Theorem \ref{theorem:noequality} we will construct $\str{A}_0$ out of some number of copies
	of a \emph{base domain}, which will be a set of elements with their $1$-types predefined.      
	 Here, the number $n_\alpha$
	 of realizations of a 1-type $\alpha$ in the base domain is obtained from the solution of $\uGamma{\bnTheta}{i,\theta}$: $n_\alpha= \max\limits_i \{ r^{E_ i}_\alpha\}$, for every $\alpha \in \AAA_{\bnTheta}$. Inequalities (E0) and (E1a) or (E1b) ensure that $n_\alpha>0$.

	Observe that the base domain can be partitioned for every $i$ into  $\phisyms{E_i}$-admissible $E_i$-classes. First, let us see that for every $i$,
		the base domain can be partitioned into $ \sum\limits_{\theta \in
			\nTheta^{E_i}} r^{E_i}_{\theta}$ disjoint parts, so that for each
		part $P$ there exists a counting type $\theta'$, safely
		extending some $\theta \in \nTheta^{E_i}$, such that the number of
		elements of every atomic 1-type $\alpha$ in $P$ equals $\theta'(\alpha)$.
		The desired partition is obtained as follows. We create
		$r^{E_i}_{\theta}$ parts for every $\theta \in \nTheta^{E_i}$. To
		each of this parts we put exactly $\theta(\alpha)$ elements of type
		$\alpha$, for every $\alpha \in \AAA_{\bnTheta}$. 
		Note, that, because of the choice
		of the numbers $n_\alpha$, we have enough copies of elements of every
		type $\alpha$. After this step we may have some elements remaining, call them \emph{redundant}.
		Observe that none of the
		types of the redundant elements is royal for $E_i$-classes.
		Indeed, if $\alpha^*$ is royal for $E_i$-classes then 
		 due to the inequalities of the form (E2)
		the value of $n_{\alpha^*}$ equals $r^{E_i}_{\alpha^*}$, and since by (E0) we have
		$r^{E_i}_{\alpha*}=\sum_{\theta \in \nTheta^{E_i}} \theta(\alpha^*) \cdot
    r^{E_i}_{\theta}$, all elements of type $\alpha^*$ from the base domain are required to build the suggested $E_i$-classes. 
		All the
	  redundant elements of type $\alpha$ are joined to a part which contains
		at least $M_\phi$ elements of $\alpha$.  Note that such a part exists due to
		inequalities of the form (E1b).
				To define the structure on each of the classes we use the fact
		that each variable $X^{E_i}_\theta$ in the system corresponds to a
		 $\phisyms{E_i}$-admissible $M_\phi$-counting type $\theta^{E_i}$
		and, (if necessary) part (ii) of Lemma \ref{l:admissibility}. Moreover, to ensure that the model is ramified 
		each of the $E_i$-classes is defined to be an $E_i$-splice of the structures given by Definition \ref{d:admissible}.

		Now, exactly as in Theorem \ref{theorem:noequality} and using the above observations, we form a ramified   $k$-dimensional grid structure $\str{A}_0\models \phisym$ that satisfies conditions (i)-(ii) of our lemma.

	Finally, we construct $\str{A}$ taking three sets of copies of $\str{A}_0$,
		each consisting of $h$ elements, where $h$ is the number of
		conjuncts of $\phi$ of the form $\cfe$,
		$A=A_0 \times \{0,1,\ldots,h-1\} \times \{0,1,2\}$; 
		we can suppose $h>0$ and provide witnesses 
		for all elements in $A$ for conjuncts of $\phi$ of the
		form $\cfe$ in a circular way, as in the proof of Theorem \ref{theorem:noequality}. 
	We argue that this is always possible. 
		
		Let us take for example an element $(a,i,j)$.
		Let $\delta_0, \ldots, \delta_{h-1}$ be all conjuncts of $\phi$ of
		the form $\cfe$, $\delta_m=\forall x (\ung(x) \rightarrow \exists y
		(\bing(x,y) \wedge \psi(x,y)))$. Let $\alpha=\type{\str{A}_0}{a}$. For $0 \le m < h$,
		if $\str{A}_0 \models \ung(a) \wedge \neg(\bing(a,a) \wedge \psi(a,a))$,
		we choose $\alpha'_m \in \AAA_{\bnTheta}$,  
		and $\beta_m \in \BBB$, whose existence
		is ensured by condition (iii) from Definition \ref{d:tupleadmissibility} of $\phi$-admissibility of $\bnTheta$. By an appropriate inequality of
		the form (E1a) or (E1b),  $\alpha_m'$ is realized in $\str{A}_0$,  say,
		by an element $b_m$. We set $\beta((a,i,j), (b_m,m,j+1 \;
		\mbox{mod} \; 3 )) :=\beta_m$. 
		The type $\beta_m$ contains $\neg E_ixy$ or $\neg E_iyx$ for every $i$, so it does not affect the equivalence relations already defined.
		
		We complete $\str{A}$
		by setting binary-free 2-types for all remaining pairs of elements. 
		This step also does not affect the equivalence classes transferred to $\str{A}$ from  $\str{A}_0$, 
	so $\str{A}$ evidently is a ramified structure that  satisfies conditions (i) and (ii) of our lemma. 
To see that $\str{A}$ is indeed a model of $\phi$ note that conjuncts of the form $\ce$ are satisfied due to
condition (ii) from Definition \ref{d:tupleadmissibility}  of $\phi$-admissibility of $\bnTheta$, conjuncts of the form $\cfes$ are satisfied due
to condition (i) from the same definition, for conjuncts of the form $\cfe$ we take care in the paragraph above,
and conjuncts of the form $\cff$ are satisfied since we use in $\str{A}$ $2$-types which are either explicitly required to be $\phi_\forall$-admissible 
or are binary-free.

 Finally, note that $|A|=3h(B|\AAA_\phi|)^k$, hence it is doubly exponential in $|\phi|$, satisfying condition (ii) of Definition~\ref{def:neat}. 
 So, $\str{A}$ is a neat model as required. 
\end{proof}

\begin{corollary}\label{lem:neat}
Every finitely satisfiable \GFtEG-sentence $\phi$ has a model of size at most doubly exponential in $|\phi|$. 
\end{corollary}
\begin{proof}
	By Lemma~\ref{lem:normalform} it suffices to consider $\phi$ in normal form. 
	Suppose $\str{A}\models \phi$, $\str{A}$ is finite. Let $\bnTheta = (\nTheta^{E_1},
\nTheta^{E_2}, \ldots, \nTheta^{E_k})$, where $\nTheta^{E_i}$ is the set of the $M_\phi$-cuttings of counting types  
realized in  $\str{A}$ by $E_i$-classes.  By Remark \ref{cor:admissible}
every element of every $\nTheta^{E_i}$ is $\phisyms{E_i}$-admissible. Additionally, fix some $\theta\in \nTheta^{E_1}$. 
Then the system 
$\Gamma_{\bnTheta}(i,\theta)$ has a non-negative integer solution (e.g.~the one corresponding to the model of $\str{A}$, cf.~Fact~\ref{f:model-to-solution}). 
	So, Lemma \ref{lem:from-solution-to-model} ensures that $\phi$ has a neat model, that in particular, has the required size. 
	\end{proof}

\begin{theorem}
	The finite satisfiability problem for \GFtEG{} is \NExpTime-complete.
\end{theorem}

\begin{proof} 
	The lower bound follows from the complexity of the fragment without equality, established in Corollary \ref{c:nexpnoeq}.
	
	By Corollary \ref{lem:neat} every finitely satisfiable \GFtEG-formula $\phi$ has a \neat model. 
	Recall that in a \neat model the number of $M_\phi$-cuttings of the counting types realized by equivalence classes 
	 is at most $m$, where $m$ is
	the number of inequalities in  $\uGamma{\bnTheta}{i,\theta}$; in turn $m$ is bounded  polynomially in the number
	of $1$-types and thus exponentially in the length of $\phi$.
  As we argue below, this allows us to check existence of a \neat model in nondeterministic exponential time. 
	
	First, we nondeterministically choose a tuple $\bnTheta$ of sets of $\phi$-admissible $M_\phi$-counting types containing at most $m$ $M_\phi$-counting types. Admissibility of each element of the guessed sets can be checked nondeterministically by guessing appropriate structures; recall that the sizes of relevant structures
	are bounded exponentially, since they contain at most $M_\phi$ realizations of every $1$-type.  
	
	Then we distinguish an $M_\phi$-counting type $\theta$  from one of these sets, say $\nTheta^{E_i}$, write the system of inequalities $\uGamma{\bnTheta}{i,\theta}$, and check if it has a 
	 non-negative integer solution. 
	 The system $\uGamma{\bnTheta}{i,\theta}$ has exponential size. 
	Since integer programming is in \NP, we can nondeterministically 
	check the existence of a non-negative integer solution of
	$\uGamma{\bnTheta}{i,\theta}$ in time polynomial w.r.t.~the size of  $\uGamma{\bnTheta}{i,\theta}$.
	
	All these gives a non-deterministic procedure working in 
	exponential time. 
	
	In fact, we could simplify the last step of the above procedure, as by Lemma \ref{lem:equations},  $\uGamma{\bnTheta}{i,\theta}$ has a non-negative integer solution if and	only if it has a non-negative rational solution. So, it suffices to look for a non-negative rational solution that can be done deterministically in time polynomial w.r.t.~the size of  $\uGamma{\bnTheta}{i,\theta}$. 
	 \end{proof}

\section{Transitive Guards}\label{sec:transitive}
 
In this section we show that every finitely satisfiable \GFtTRG{}
formula $\phi$ has a model of size at most doubly
exponential in $|\phi|$ and that the finite satisfiability problem
for \GFtTRG{} is \TwoExpTime-complete.

\subsection{Basic notions and outline of the bounded size model construction}

Recall that for a given transitive and reflexive special symbol $T$, a $T$-clique in a structure $\str{A}$
is a maximal set $C \subseteq A$, such that for all $a,b \in C$ we have $\str{A} \models Tab \wedge Tba$; if $C$ is a $T$-clique in $\str{A}$ then $\str{C}=\str{A}\restr C$ is a $T$-class.    

In the previous section we considered counting types 
of equivalence classes. In this section we will analogously work with counting
types of classes formed by transitive cliques.
 Moreover, to be able to provide witnesses for conjuncts of the form $\cfes$ 
with non-symmetric special guards, we enrich the notion of the counting type 
with two subsets of 1-types, $\cal A$ and $\cal B$,  corresponding
to 1-types of elements located in $\str{A}$ {\em above},
respectively {\em below}, the elements of the clique. 

\begin{definition}      
An {\em enriched counting type} is a tuple $\otheta=(\theta, \cal{A}, \cal{B})$, where $\theta$ is a counting type and $\cal{A}, \cal{B}$ are
sets of $1$-types.
We say
	that a $T$-class $\str{C}$ in a structure $\str{A}$   \emph{realizes} (or \emph{has}) an enriched counting
	type $\otheta=(\theta, {\cal A}, {\cal B})$ if:    
	\begin{enumerate}[(i)]
		\item
		$\ctype{\str{C}}=\theta$, 
		\item
		${\cal A} = \{\alpha  \mid \exists  a \in C, b \in A \mbox{ such that } \;
		\type{\str{A}}{b}=\alpha  \wedge \str{A} \models Tba \wedge \neg Tab
		\},$
		\item
		${\cal B} = \{\alpha  \mid \exists  a \in C, b \in A \mbox{ such that } \;
		\type{\str{A}}{b}=\alpha  \wedge \str{A} \models Tab \wedge \neg Tba
		\}$.
	\end{enumerate}     
	In this case we also say that the enriched counting type $\otheta$ {\em enriches} the counting type $\theta$.
	Enriched counting types whose $\theta$-components are $n$-counting types are called $\emph{enriched}$ $n$-\emph{counting types}.
	For an enriched counting type $\otheta=(\theta, {\cal A}, {\cal B})$, the {\em cutting of} $\otheta$ to $n$ (or \emph{the} $n$-\emph{cutting} of $\otheta$) is the 
		enriched $n$-counting type $\otheta'=(\theta', \cal{A}, \cal{B})$, where $\theta'$  is the cutting of $\theta$ to $n$. 
\end{definition}

Similarly to the previous section, we usually work in contexts in which a normal form $\phi$ is fixed and then we are interested in 
$M_\phi$-cuttings of enriched counting types,     
 where  $M_\phi=3|\AAA_\phi||\phi|^3$.
Before introducing the next notion let us recall the definition of $\phisyms{T}$ and introduce one more fragment of the formula $\phi$ denoted by $\phifull{T}$ below. 
\begin{eqnarray*}
	\phisyms{T} &:=& \phiuniv \wedge \bigwedge\{\psi \mid ~ \psi \mbox{ is a conjunct of type $\cfes$ where }  \eta(x,y) = Txy\wedge Tyx\}\\ 
	\phifull{T} &:=& \phiuniv \wedge \bigwedge\{\psi \mid ~ \psi \mbox{ is a conjunct of type $\cfes$ where }  \eta(x,y) \vdash Txy \vee Tyx\} 
\end{eqnarray*}
To simplify notation we  occasionally allow ourselves to treat a counting type $\theta$ as a set consisting of those $1$-types which appear in $\theta$; this  applies in particular to expressions like $\alpha \in \theta$ or $\cal{A} \cup \theta$, for $\cal{A} \subseteq \AAA_\phi$.

\begin{definition} \label{d:tadmiss}
	We say that an enriched counting type      
	$\otheta=(\theta, \cal{A}, \cal{B})$   is
	$\phifull{T}$-admissible   if 
	\begin{enumerate}[(i)]
		\item $\theta$ is  $\phisyms{T}$-admissible     
		in the
		sense of Definition \ref{d:admissible}, 
		\item \begin{enumerate}[(a)]
		\item for any $\alpha \in \theta$ and  $\alpha' \in \cal{B}$
		there is a $\phi_\forall$-admissible $2$-type $\beta$ such that $\beta(x,y) \models Txy \wedge \neg Tyx \wedge \alpha(x) \wedge \alpha'(y)$; 
		\item analogously,  for any $\alpha \in \theta$ and  $\alpha' \in \cal{A}$ there is a $\phi_\forall$-admissible  $2$-type $\beta$ such that $\beta(x,y) \models  \neg Txy \wedge Tyx \wedge \alpha(x) \wedge \alpha'(y)$;
		\end{enumerate}
		\item 
		\begin{enumerate}[(a)]
			\item for every $\alpha \in \theta$, if $\alpha(x) \models \ung(x)$ then for every conjunct of $\phifull{T}$  of the form $\cfes$ of the shape 
			$\forall x (\ung(x) \rightarrow \exists y (T xy \wedge \neg T yx \wedge \psi(x,y)))$ 
			there is $\alpha' \in \cal{B}$ and a $\phi_\forall$-admissible $2$-type $\beta$
			such that $\beta(x,y) \models Txy \wedge \neg Tyx \wedge \psi(x,y) \wedge \alpha'(y)$;
			\item
			analogously, for every $\alpha \in \theta$, if $\alpha(x) \models \ung(x)$ then for every conjunct of the form $\cfes$  of the shape $\forall x (\ung(x) \rightarrow \exists y (T yx \wedge \neg T xy \wedge \psi(x,y)$
			there is $\alpha' \in \cal{A}$ and a $\phi_\forall$-admissible $2$-type $\beta$
			such that $\beta(x,y) \models Tyx \wedge \neg Txy \wedge \psi(x,y) \wedge \alpha'(y)$.
		\end{enumerate}
	\end{enumerate}
\end{definition}

\medskip

Let $\str{A}$ be a finite model of $\phi$. Our plan is to extract from $\str{A}$ 
a \emph{certificate for finite satisfiability of} $\phi$, collect its
important properties, and argue that these properties allow us to build a ramified 
model $\str{B} \models \phi$ of bounded size.

One of the crucial steps will be an application of some results obtained in the case of \GFtEG{}. Indeed, if we remove from $\str{A}$  non-symmetric transitive connections
then, in the obtained structure $\hat{\str{A}}$, transitive and reflexive relations
become equivalences and transitive cliques behave
like equivalence classes.
It can be easily seen that $\hat{\str{A}} \models \phieq$. 
For every enriched 
counting type $\otheta=(\theta, \cal{A}, \cal{B})$ realized in $\str{A}$ by a $T$-class we use 
Lemma
\ref{lem:from-solution-to-model}  to produce a neat  model of
$\phieq$  containing a distinguished $T$-class of counting type $\theta'$ such that the $M_\phi$-cuttings
of $\theta$ and $\theta'$ are identical.
These structures are the building blocks for the construction
of $\str{B}$. We arrange some number of copies of these structures on a cylindrical surface, and provide non-symmetric transitive witnesses in a regular manner, enforcing the enriched types of the distinguished $T$-classes to have a form $\otheta'=(\theta', {\cal A}', {\cal
	B}')$, for some $\cal{A}'$, $\cal{B}'$ such that ${\cal A}' \subseteq {\cal A}$ and ${\cal B}'
\subseteq {\cal B}$.    

Below we describe the construction in detail.

\subsection{Certificates of finite satisfiability}

Below we define the notion of a \emph{certificate of finite satisfiability} for a normal form \GFtTRG{} formula $\phi$ and  we observe that such a certificate can be easily extracted from an existing model of $\phi$. 

\begin{definition} \label{d:cert}
	We say that a tuple $(\boTheta, \cal{F})$,
	where
	\begin{itemize}
	\item
	$\boTheta = ( \oTheta^{T_1}, \ldots, \oTheta^{T_k} )$ is a list of non-empty sets of enriched $M_\phi$-counting types, and
	\item $\cal{F} = (\cal{F}_1, \ldots, \cal{F}_k)$ is a list of functions, such that for each $\otheta \in \oTheta^{T_i}$, $\cal{F}_i$  returns
	a finite structure  $\str{F}^{T_i}_{\otheta}$,
	\end{itemize}
	is \emph{a certificate} of
	(or \emph{certifies}) finite satisfiability for a normal form \GFtTRG{} formula $\phi$ if the following conditions hold:
	\begin{enumerate}[(i)]
		\item  for each $ \otheta=(\theta, \cal{A}, \cal{B})  \in \oTheta^{T_i}$: 
		\begin{enumerate}
			\item  $\otheta{}$  is $\phifull{T_i}$-admissible  
			\item  for every $\alpha \in \cal{B}$ there is $\otheta'=(\theta', \cal{A}', \cal{B}')$ in $\oTheta^{T_i}$ such that:
			\begin{itemize}
				\item  $\alpha \in \theta'$; $\alpha \not\in \cal{B}'$; 
				\item $\cal{B}' \cup \theta' \subseteq \cal{B}$;
				$\cal{A} \cup \theta \subseteq \cal{A}'$ 
			\end{itemize}
			\item  symmetrically, for every $\alpha \in \cal{A}$ there is $\otheta'=(\theta', \cal{A}', \cal{B}')$ in $\oTheta^{T_i}$ such that:
			\begin{itemize}
				\item  $\alpha \in \theta'$; $\alpha \not\in \cal{A}'$, 
				\item $\cal{B} \cup \theta \subseteq \cal{B}'$;
				$\cal{A}' \cup \theta' \subseteq \cal{A}$ 
			\end{itemize}
		\end{enumerate}
		\item for every $\str{F}^{T_i}_{\otheta}$:
		\begin{enumerate}
			\item all special symbols are interpreted as equivalences,
			\item $\str{F}^{T_i}_{\otheta}$  is a neat model of $\phieq$ (cf.~Definition \ref{def:neat}), 
			\item for every special $T_j$, for every $T_j$-class $\str{C}$ in $\str{F}^{T_i}_{\otheta}$, there is $\otheta' \in \oTheta^{T_j}$  enriching the $M_\phi$-cutting of $\ctype{\str{C}}$,    
		  \item  assuming $\otheta=(\theta, \cal{A}, \cal{B})$ for some $\cal{A}$, $\cal{B}$, there is a $T_i$-class in $\str{F}^{T_i}_{\otheta}$ whose counting type cut to  $M_\phi$ equals $\theta$.     
		\end{enumerate}
	\end{enumerate}
	
\end{definition}

The following observation follows directly from the definition.

\begin{claim} \label{c:tsize}
For every certificate $(\boTheta, \cal{F})$ the size of each $\oTheta^{T_i}$ and the sizes of structures returned by  each $\cal{F}_i$  are bounded doubly exponentially in the size of the signature. Thus the  size of a description of a certificate is  bounded doubly exponentially in the size of the signature.   
\end{claim}

Let us see how, given a finite model of $\phi$,  to construct a tuple $(\boTheta, \cal{F})$ certifying finite satisfiability of $\phi$.

\begin{lemma}\label{l:certexist}
	If $\phi$ is a finitely satisfiable normal form \GFtTRG{} formula then it has a certificate for finite satisfiability.
\end{lemma}
\begin{proof}
	Let $\str{A}$ be a finite model of $\phi$. 
	We first construct $\boTheta=(\oTheta^{T_1}, \ldots, \oTheta^{T_k})$.
	Let $\oTheta^{T_i}$ be the set of the $M_\phi$-cuttings of the enriched counting types of $T_i$-classes of $\str{A}$.      
	Let us take any $\otheta \in \oTheta^{T_i}$, $\otheta=(\theta, \cal{A}, \cal{B})$. It is clear that $\otheta$ is $\phifull{T_i}$-admissible. Consider condition (i)(b) of
	Definition \ref{d:cert} and take any $\alpha' \in \cal{B}$. Choose any $T_i$-class $\str{C}$ of type $\otheta$ from $\str{A}$.
	Since $\alpha' \in \cal{B}$ there is an element $a' \in A$ of type $\alpha'$ such that for any $a \in C$ we have $\str{A} \models Taa' \wedge \neg Ta'a$.     
	Finiteness of $\str{A}$ guarantees that there is a minimal such $a'$ (that is such $a'$ that for any $a''$ of $1$-type $\alpha'$ we have $\str{A} \not\models Ta'a'' \wedge \neg Ta''a'$). Let $\otheta'=(\otheta, \cal{A}', \cal{B}')$ be the $M_\phi$-cutting of the enriched  
	counting type	 of the $T_i$-class
	of $a'$.  Obviously $\otheta' \in \oTheta^{T_i}$. It is readily verified that $\otheta'$ is as required, in particular $\alpha' \not\in \cal{B}'$
	due to the minimality of $a'$. 
	Analogously we can check that condition (i)(c) is satisfied.
	
	Let us now construct $\cal{F}=(\cal{F}_1, \ldots, \cal{F}_k)$. 
	Consider the structure $\hat{\str{A}}$---a modification of $\str{A}$ in
	which, in all 2-types, for all transitive symbols $T$ we substitute $Txy \wedge \neg Tyx$ or $Tyx
	\wedge \neg Txy$ with $\neg Txy \wedge \neg Tyx$, for all
	transitive $T$. This way, each $T$ is an equivalence relation in $\hat{\str{A}}$. 
	Obviously, $\hat{\str{A}} \models \phieq$.
	Moreover, since $\phieq$  
	uses only
	symmetric  guards in conjuncts of type $\cfes$, it may be treated as a \GFtEG{} formula.
	For each special symbol $T$,  let $\nTheta^{T}$ be the set of the $M_\phi$-cuttings of counting types 
	 realized in
	$\hat{\str{A}}$ by the $T$-classes, and let $\bnTheta=(\nTheta^{T_1},\ldots,\nTheta^{T_k})$. 
	For each $T_i$, for each $\otheta \in \oTheta^{T_i}$, $\otheta=(\theta, \cal{A}, \cal{B})$, consider
	the system of inequalities
	$\Gamma_{\bnTheta}(i,\theta)$, as defined in Section \ref{sec:eq-with-equality}.
	Note, that 	$\Gamma_{\bnTheta}(i,\theta)$ has a non-negative solution: the one
	corresponding to the model $\hat{\str{A}}$. Take a neat model guaranteed now by  Lemma \ref{lem:from-solution-to-model}  as $\str{F}_{\otheta}^{T_i}=\cal{F}_i(\otheta)$.
		It is readily verified that it meets the required conditions (ii)(a)-(ii)(d). In particular condition (ii)(c) is satisfied since the $M_\phi$-cuttings of all counting types 
		of $T_i$-classes realized in
	$\str{F}_{\otheta}^{T_i}$ are members of $\nTheta^{T_i}$, and the counting types from $\nTheta^{T_i}$ are obtained just by dropping the components $\cal{A}$ and $\cal{B}$ from  enriched types from $\oTheta^{T_i}$. Condition (ii)(d) is satisfied due to inequality (E3). 
	\end{proof}

\subsection{Construction of a bounded size model}

Now we show that given a certificate of finite satisfiability for a normal form \GFtTRG{} formula $\phi$ we can build its bounded finite model.

\begin{lemma} \label{l:modelt}
	If $(\boTheta, \cal{F})$ certifies finite satisfiability of $\varphi$ then $\varphi$ has a finite model.
\end{lemma} 

The rest of this section is essentially devoted to the proof of the above lemma. 

Let $(\boTheta, \cal{F})$ be a certificate of finite satisfiability of a normal form \GFt{}-formula $\phi$. 
We define an auxiliary choice function $cc$ which for every $\str{F}^T_{\otheta}$, assuming $\otheta=(\theta, {\cal A}, {\cal B})$,  
returns one of its $T$-classes whose counting type cut to $M_\phi$
equals $\theta$. The role of $cc(\str{F}^T_{\otheta})$ will be to provide non-symmetric witnesses for some elements external to $\str{F}^T_{\otheta}$. 
We define another auxiliary function $enr$ which for each structure $\str{F}^T_{\otheta}$,   and  each $T'$-class $\str{C}$ of $\str{F}^T_{\otheta}$ returns
an enriched $M_\phi$-counting type from $\oTheta^{T'}$, enriching  the $M_\phi$-cutting of $\ctype{\str{C}}$.     
Namely, if $T'=T$ and $\str{C}=cc(\str{F}^T_{\otheta})$ 
 then $enr(\str{C})=\otheta$;      
 otherwise we choose an arbitrary type meeting 
the requirement. At least one appropriate type exists in $\oTheta^{T'}$ due to condition (ii)(c) of Definition \ref{d:cert}.
The purpose of $enr(\str{C})$ is to say, elements of which 1-types are allowed to be connected to the $T'$-clique  $C$ by the transitive relation 
$T'$ 
from above and from below. We remark that (though it would be possible) we will not try to make the $M_\phi$-cutting 
of the enriched type of a class 
$\str{C}$ in our final model to be exactly equal to $enr(\str{C})$, but rather, if $enr(\str{C})=(\theta, {\cal A}, {\cal B})$ then we make it equal
$\otheta=(\theta, {\cal A}', {\cal
	B}')$, for some $\cal{A}'$, $\cal{B}'$ such that ${\cal A}' \subseteq {\cal A}$ and ${\cal B}'
\subseteq {\cal B}$. 

We split the construction of the desired model $\str{B} \models \phi$ into a few steps.

\medskip\noindent
{\bf Universe.}
Let $K=2|\AAA_\phi|+1$. Let $m^i_l$ be the number of conjuncts of type $\cfes$ with guards  $T_ixy \wedge \neg T_iyx$,
and $m^i_u$ the number of conjuncts of type $\cfes$ with guards $T_iyx \wedge \neg T_ixy$. Let $m^i=\max \{m^i_l, m^i_u\}$ 
and let  $m=\max_i m^i$. 
We define the universe of $\str{B}$ to
be
 $$B=\{ 0,\ldots, K-1 \} \times
\{0,\ldots, 3 \} \times \{1, \ldots, m \} \times
{\bigcup\limits^{{\bm \cdot}}_{1\le i \le k; \otheta \in \oTheta^{T_i}}}
F_{\otheta}^{T_i}$$ 

We can think that $B$ consists of $K$
rows and $4$ columns, and that the intersection of each row and
each column contains $m$ disjoint copies of each $\str{F}^T_{\otheta}$.
We imagine that row $K-1$ is glued to row $0$, and thus  a
cylindrical surface is obtained.

Initially, we impose appropriate structures on the copies of
each  $\str{F}^T_{\otheta}$. For all $i, j, l$, and $a \in  F^T_{\otheta}$ let us set
$\type{\str{B}}{(i,j,l,a)} := \type{\str{F}^T_{\otheta}}{a}$, and for $a, b \in   {F}^T_{\otheta}$, 
$\type{\str{B}}{(i,j,l,a),(i,j,l,b)}:=\type{ \str{F}^T_{\otheta}}{a,b}$.
We also transfer the values of the functions $cc$ and $enr$ from the structures  $\str{F}^T_{\otheta}$ 
to their copies.

\medskip\noindent
{\bf Non-symmetric witnesses.}
Now we provide, for all elements, witnesses for all conjuncts of type $\cfes$ with non-symmetric  guards. 
  This is done in a regular manner. Elements from row $i$
find their lower witnesses, i.e., witnesses for conjuncts with $\eta(x,y)=Txy \wedge \neg Tyx$, in the row $i+1 (\mbox{mod} \; K)$, and 
their upper witnesses, i.e., witnesses for conjuncts with $\eta(x,y)=Tyx \wedge \neg Txy$, in the row $i-1 (\mbox{mod} \; K)$. Elements
from column $j$ look for their lower witnesses in column $lower(j)$
and for their upper witnesses in column $upper(j)$, where $lower, upper:
\{0,1,2,3 \} \rightarrow \{0,1,2,3 \}$ are defined as follows: 
$lower(0)=lower(2)=0$, $lower(1)=lower(3)=1$, $upper(0)=upper(3)=3$, $upper(1)=upper(2)=2$. 

This
strategy guarantees two important properties, which will allow us to carry our construction out without conflicts. First, there will be
no pair of elements $a$, $b$ such that $a$ use $b$ as a lower
witness and $b$ use $a$ as an upper witness (since both $upper \circ lower$ and
$lower \circ upper$ have no fixed points). Second, none of the cliques will be
used as a source for both a lower witness and an upper witness
(since $upper$ and $lower$ have disjoint images).

Let us consider explicitly the case of, say, providing lower
$T$-witnesses for element $a=(5,1,3,a_0)$, where $a_0 \in
F^{T'}_{\otheta'}$. 

Let $\alpha$ be the $1$-type of $a_0$.
Let $C$ be the $T$-clique of $a_0$ in
$\str{F}^{T'}_{\otheta'}$. Let $\otheta=(\theta, \cal{A}, \cal{B})=enr(\str{C})$. The outlined
strategy says that lower witnesses should be looked for in row 6,
column 1. Let $\delta_1, \ldots, \delta_l$ be the list of the conjuncts of $\phi$ of
the form $\cfes$  guarded by $Txy \wedge \neg Tyx$, 
$\delta_i =\forall x (\ung_i(x) \rightarrow \exists y (T xy \wedge \neg T yx \wedge \psi_i(x,y)))$.     
Note that $l \le m$. 

For each $1 \le i \le l$, if $\alpha(x) \models \ung_i(x)$ then take $\alpha_i' \in \cal{B}$ and $\beta_i$ guaranteed by condition (iii)(a) of Definition \ref{d:tadmiss}.
Let $\otheta_i'=(\theta_i', \cal{A}'_i, \cal{B}'_i) \in \oTheta^T$ be the enriched $M_\phi$-counting type guaranteed by condition (i)(b) of Definition \ref{d:cert} for $\alpha_i'$, i.e., 
$\alpha_i' \in \theta'_i$, $\alpha'_i \not\in \cal{B}'_i$, $\cal{B}'_i \cup \theta' \subseteq \cal{B}$ and $\cal{A} \cup \theta \subseteq \cal{A}'_i$. 
Choose an element $a'_i$ of $1$-type $\alpha'_i$ in the $T$-class 
$cc(\{ 6 \}   \times  \{1\}  \times   \{i\} \times \str{F}^T_{\otheta'_i})$. 
     
Set $\type{\str{B}}{a, a_i'}:=\tsplice{\beta_i}{T}$. 

In a similar way, in accordance with the outlined strategy, we provide lower and upper witnesses for all elements in $\str{B}$.

\medskip\noindent
{\bf Transitive closure.}
During the previous step of providing witnesses we defined 2-types
containing non-symmetric transitive connections between some pairs
of elements belonging to consecutive rows. 

Let  $C$, $C'$ be a pair of $T$-cliques in 
$\str{B}$. We write $C' <_T^{\str{B}} C$ if for all $a \in C'$, $b
\in C$, $\str{B} \models   Tab \wedge \neg Tba$. 
We  say that there exists a $T$-path from 
$C'$ to  $C$ in (the current, non-fully specified, version of) $\str{B}$ if there exists a sequence
of  $T$-cliques $C=C_0,C_1, \ldots C_{l-1}, C_l=C'$ such that for
every pair $C_i$, $C_{i+1}$ we defined a lower $T$-witness for an
element from $C_i$ in $C_{i+1}$ or an upper $T$-witness for an element from $C_{i+1}$ in $C_i$.

In order to complete the definition of $T$ we have to ensure that for every transitive $T$, for every pair of $T$-cliques $C$, $C'$ in $\str{B}$, if there is a $T$-path from $C'$ to $C$, then $C' <_T^{\str{B}} C$. 

Let us note basic properties of the above defined $T$-paths in (the current version of) $\str{B}$. 

\begin{claim} \label{c:shape} For every transitive symbol $T$, a $T$-path from a $T$-clique $C$ to a $T$-clique $C'$ has the form
$C_1, \ldots, C_s, C_{s+1}, \ldots, C_l$  (one of the fragments $C_1, \ldots, C_s$ or $C_{s+1}, \ldots, C_l$ may be empty), 
where 
\begin{enumerate}[(i)]
\item two consecutive cliques $C_i, C_{i+1}$ belong to rows $j$, $(j+1) \mod K$, for some $j$,
\item $C_1, \ldots, C_s$ belong to the same column ($2$ or $3$); similarly, $C_{s+1}, \ldots, C_{l}$ belong to the same column ($0$ or $1$),
\item each of the cliques from the considered $T$-path, except for possible at most one of  $C_s, C_{s+1}$, belongs to a copy of 
$\str{F}^T_{\otheta}$ (for some $\otheta$); one of $C_s, C_{s+1}$ may belong to a copy of  $\str{F}^{T'}_{\otheta}$ (for some $\otheta$ and $T' \not=T$), 
\item $s \le  |\AAA_\phi|+1$, $l-s \le  |\AAA_\phi|+1$. Moreover, $l \le 2 |\AAA_\phi| + 1$.    
\end{enumerate}
\end{claim}
\begin{proof}
Conditions (i)-(iii) follow directly from our strategy of finding non-symmetric witnesses (in the next row, in the column given by functions $upper$ and $lower$, in 
a copy of an appropriate $\str{F}^{T}_{\otheta}$ structure).
For (iv) let $\otheta_i=(\theta_i, \cal{A}_i, \cal{B}_i)= enr(\str{C}_i)$. Due to our strategy, for $i=1, \ldots, s-1$ we have $\cal{A}_i \subsetneq \cal{A}_{i+1}$.
Similarly, for $i=s, \ldots, l-1$ we have $\cal{B}_{i+1} \subsetneq \cal{B}_i$. Since $|\cal{A}_i|, |\cal{B}_i| \le |\AAA_\phi|$ the claim about the length
of each of the two fragments follows. The claim about the length of the whole path follows from the observation that also one of 
the containments $\cal{A}_s \subseteq \cal{A}_{s+1}$ or $\cal{B}_{s+1} \subseteq \cal{B}_{s}$ must be strict, since either
one of the elements of $C_{s+1}$ is a $T$-witness for an element of $C_s$, or the other way round.
\end{proof}

Suppose now that $C$ and $C'$ are $T$-cliques in $\str{B}$ and there exists a $T$-path, for some transitive $T$, from $C$ to $C'$,
$C=C_1, C_2, \ldots, C_{l-1}, C_l=C'$. 
We describe below how to make $C' <^{\str{B}}_T C$. 

$C' <^{\str{B}}_T C$: For every $a' \in C'$, $a
\in C$, if the 2-type for $a'$ and $a$ is not defined, then let
$\alpha$ be the $1$-type of $a$ and $\alpha'$ the $1$-type of $a'$.
Let $\otheta=(\theta, \cal{A}, \cal{B})=enr(\str{C})$, $\otheta'=(\theta', \cal{A}', \cal{B}')=enr(\str{C}')$.
Observe that $\theta' \subseteq \cal{B}$. Indeed, denoting
$\otheta_i=(\theta_i, \cal{A}_i, \cal{B}_i) = enr(\str{C}_i)$, we have 
by our construction $\theta'=\theta_l \subseteq \cal{B}_{l-1} \subseteq \ldots, \cal{B}_1 =\cal{B}$.
We have that $\alpha \in \theta$, and, by the remark above, $\alpha' \in \cal{B}$.
We take a $2$-type $\beta$ guaranteed by condition (ii) of Definition \ref{d:tadmiss} and
set $\type{\str{B}}{a,a'}:= \tsplice{\beta}{T}$. 

\smallskip
Now we use the above procedure for every transitive $T_i\in \sigma$, and for every $T_i$-cliques $C$, $C'$ in $\str{B}$, such that  there is a $T_i$-path from $C$ to $C'$. 
This step ensures that the interpretation of each $T_i$ is transitive in $\str{B}$. 
The following claim shows that setting $2$-types in the above step can be done without conflicts.

\begin{claim} \label{claimokc}
 Let  $a \in B$ belongs to a $T$-clique $D$, $b \in B$ belongs to a $T$-clique $D'$ and there is a $T$-path from $D$ to
		$D'$. Let  $T' \not =T$, let $C$ be the $T'$-clique of $a$ and $C'$ the $T'$-clique of $b$. Then  there is no $T'$-path from $C$ to $C'$ nor
		from $C'$ to $C$.
		\end{claim}
	
	\begin{proof}
		Assume first that  there is a $T'$-path from $C$ to $C'$. We consider two cases. In the first, let the $T$-path from $D$ to $D'$ stay
	in the same column. Then the $T'$-path from $C$ to $C'$ must also stay in this column. If it is column $0$ or $1$ then 
	$C'$ must belong to a copy of structure $\str{F}^{T'}_{\otheta'}$ for some $\otheta'$ and $D'$ - to a copy of structure $\str{F}^{T}_{\otheta}$ for some $\otheta$. 
	But these copies  are disjoint and $C'$ and $D'$ share the element $b$; contradiction. 
	Analogously if this column is $2$ or $3$ then $C$ must belong to a copy of structure $\str{F}^{T'}_{\otheta'}$ and $D$ - to a copy of structure $\str{F}^{T}_{\otheta}$; contradiction.
	Assume now that the $T$-path from $D$ to $D'$ changes the column. As an example consider a subcase in which  it starts in column $2$ and ends in column $0$. 
	Then the $T'$-path from $C$ to $C'$ also starts in column $2$ and ends in column $0$. Both paths change the column for the same reason: one of
	the elements from column $2$ looks for its lower witness in column $0$. Thus all cliques from the fragment of the $T$-path from column $0$ must belong to 
	copies of structures $\str{F}^T_{\otheta}$ for some types $\otheta$ and  all cliques from the fragment of the $T'$-path from column $0$ must belong to copies of structures 
	$\str{F}^{T'}_{\otheta'}$ for some types $\otheta'$. Again, it follows that $C'$ and $D'$ are disjoint; contradiction. The other subcases can be treated analogously.

	Observe now that  there is no $T'$-path from $C'$ to $C$. 
	Assume to the contrary that such a path exists. From Claim \ref{c:shape}
	it follows that both this path and the $T$-path from $D$ to $D'$ must stay in the same column. From the same claim we have then
	that each of these paths consists of at most $|\AAA_\phi| + 1$ cliques. This is however not enough to round our whole structure consisting of $K=2|\AAA_\phi| +1$ rows. 
 (Recall that a single step on each of the paths goes from row $i$ to row $i+1 \mod K)$. 
		\end{proof}

\medskip\noindent
{\bf Remaining 2-types}
	For every pair of elements $a,a' \in B$ for which we have not
	defined a 2-type yet, we set it to be the unique  binary free type
	containing the 1-types of $a$ and $a'$. This step completes the
	construction of the structure $\str{B}$.

\medskip
Let us now see that indeed:
	\begin{claim}
	$\str{B} \models \phi$.
	\end{claim}
	\begin{proof}
	First note that all special symbols are indeed interpreted as transitive and reflexive relation: they are equivalences when 
	restricted to copies of $\str{F}_{\otheta}^T$ structures, the non-symmetric connections appearing in step \emph{Non-symmetric
	witnesses} are transitively closed in step \emph{Transitive closure}. We take care for witnesses for conjuncts of type $\ce$, $\cfe$, and 	conjuncts of type $\cfes$ with symmetric special guards when the structures $\str{F}_{\otheta}^T$ are built (recall that each $\str{F}_{\otheta}^T \models \phieq$),  and for witnesses for conjuncts of type $\cfes$ with
	non-symmetric special guards in step \emph{Non-symmetric
	witnesses}. The universal conjuncts of type $\cffs$ and $\cff$ are satisfied since any $2$-type appearing in $\str{B}$ is $\phi_\forall$-admissible (it is either a
	$2$-type realized in a copy of one of  the structures $\str{F}_{\otheta}^T$, a $2$-type guaranteed by parts (ii)-(iii) of Definition \ref{d:tadmiss}, or a binary-free type).
		\end{proof}
	
	This finishes the proof of Lemma \ref{l:modelt}. 
	Combining it with Lemmas \ref{lem:normalform}, \ref{l:certexist} and Claim \ref{c:tsize} we get the
	following tight bound on the size of finite models.
		
	\begin{theorem} \label{t:sizet}
		Every finitely satisfiable \GFtTRG{} formula $\phi$ has a model of
		size at most doubly exponential in $|\phi|$.     
	\end{theorem}

We remark that our construction yields ramified models. Indeed, structures $\str{F}_{\otheta}^{T_i}$ are
	constructed as models for \GFtEG{} formulas in Section \ref{sec:eq-with-equality} which were ramified.
	In steps \emph{Non-symmetric witnesses} and \emph{Transitive closure} we use $T$-splices of $2$-types, 
	and in step \emph{Remaining $2$-types} we use binary-free types.

	\subsection{Complexity}
	
	Theorem \ref{t:sizet} immediately gives \TwoNExpTime-upper bound on the complexity
	of the finite satisfiability problem. We however work further
	to obtain the tight complexity bounds.
	
	\begin{theorem}\label{t:TRGcomp}
		The finite satisfiability problem for \GFtTRG{} is \TwoExpTime-complete. 
	\end{theorem}
	
	The lower bound can be
	obtained in the presence of just one transitive symbol, as
	shown in  \citeN{Kie06} where the  (unrestricted) satisfiability problem is considered. 
The proof  there  is by an encoding of alternating Turing machines working in
	exponential space. Actually, it involves  infinite tree-like models, since the
	machines are allowed to loop and thus to work infinitely. It is however easy
	to adapt the proof to the case of finite models: it suffices to consider only
	machines which stop after doubly exponentially many steps, which can then
	be naturally encoded by finite trees. 
		\begin{remark}\label{remark:partial} 
				The intended models 	in the proof from \citeN{Kie06} do not contain non-singleton cliques, thus, 
	the proof
	works also for \GFt{} with one partial order and the (finite) satisfiability problem for \GFtPG{} is \TwoExpTime-hard. 
		\end{remark}

	To justify the upper bound we design an alternating algorithm working in exponential space.
	The $\TwoExpTime$-bound follows then from a well known result that $\TwoExpTime=\textsc{AExpSpace}$ \cite{CKS81}. In the description of the algorithm we  use short phrases of the form \emph{guess an object $A$ such that condition $B$ holds.} In a more formal description it would be replaced with \emph{guess an object $A$ of the appropriate shape, verify that it meets condition $B$ and reject if $B$ is not satisfied}.
	
	\medskip\noindent
	{\bf Algorithm} {\tt GF$^2{+}$TG{}-FINSAT}. Input: a normal form \GFtTG{} formula $\phi$. 
	
	\begin{enumerate}
		
		\item $Counter:= |\{(i,\otheta) \mid 1 \le i \le k \mbox{ and } \otheta \mbox{ is an enriched $M_\phi$-counting type} \}|$ 
		 	\item 
		{\bf guess} a $\phifull{T_1}$-admissible enriched $M_\phi$-counting type $\otheta^*=(\theta^*, \cal{A}^*, \cal{B}^*)$; set $s:=1$
		\item {\bf guess} a tuple of sets of $M_\phi$-counting types $\nTheta=(\nTheta^{T_1}, \ldots, \nTheta^{T_k})$, each of the sets       
		of size  bounded exponentially
		by $\mathfrak{f}(|\phi|)$ (for the function $\mathfrak{f}$ from Definition \ref{def:neat}), 
				such that
		\begin{itemize}
			\item $\theta^* \in \nTheta^{T_s}$
					\item every element of $\nTheta^{T_i}$ is $\phisyms{T_i}$-admissible
			\item the system of inequalities $\Gamma_{\bnTheta}(s,\theta^*)$  
			 has an integer solution in which all the unknowns have non-zero values.
		\end{itemize}
		\item {\bf universally choose} either to stay in $\theta^*$ or to move to some other counting type appearing in $\nTheta$. In the latter case:
		\begin{enumerate} 
			\item {\bf universally choose} $s \in \{1, \ldots, k \}$ and  $\theta^*$ appearing in $\nTheta^{T_s}$ different from the current $\theta^*$ 
			\item {\bf guess}  $\cal{A}^*$ and $\cal{B}^*$ such that  $\otheta^*=(\theta^*, \cal{A}^*, \cal{B}^*)$ enriches $\theta^*$ and
			is $\phifull{T_s}$-admissible 
		\end{enumerate}
		\item {\bf universally choose} to go Down or Up
		\item if Down has been chosen then if $\cal{B}^*=\emptyset$ then {\bf accept}; otherwise {\bf universally choose} a $1$-type $\alpha \in \cal{B}^*$   and {\bf guess} an enriched $M_\phi$-counting type 
		$\otheta'=(\theta', \cal{A}', \cal{B}')$ such that:
		\begin{itemize}
			\item $\otheta'$ is $\phifull{T_s}$-admissible 
			\item $\alpha \in \theta'$; $\alpha \not\in \cal{B}'$
			\item $\cal{B}' \cup \theta' \subseteq \cal{B}^*;  \cal{A}^* \cup \theta^* \subseteq \cal{A}'$
					\end{itemize}
		if Up has been chosen then proceed symmetrically
		\item $\otheta^*:=\otheta'$
		\item $Counter:=Counter-1$; if $Counter =0$ then {\bf accept}; else goto Step (3)
		
	\end{enumerate}
	
	It is not difficult to see that the algorithm uses only exponential space. Indeed, the value of the variable $Counter$ is bounded
	doubly exponentially, thus it can be written using exponentially many bits. Also, an (enriched) $M_\phi$-counting type can be described using exponentially many bits, and the algorithm needs to store at most exponentially many such types.
	Let us now argue that the output of the algorithm is correct.
	
	\begin{claim}
		The Algorithm {\tt GF$^2{+}$TG-FINSAT} accepts its input normal form \GFtTG{} formula $\phi$ iff $\phi$ has a finite model.
	\end{claim}
	\begin{proof}
		Assume that $\phi$ has a finite model. We show how the guesses in our algorithm can be made to obtain an accepting run on $\phi$.
		By Lemma \ref{l:certexist} there is a certificate $(\boTheta, \cal{F})$ for finite satisfiability of $\phi$, $\boTheta=(\oTheta^{T_1}, \ldots, \oTheta^{T_k})$.
		In Step (2) we choose $\otheta^*$ to be any element of $\oTheta^{T_1}$.  We next explain that the following loop invariant is       
		ensured for the loop (3)-(8): whenever it is entered
		with $s$ and $\otheta^*$ such that $\otheta^{*} \in \oTheta^{T_s}$ then we can go through it and either accept or 
		reach Step (8) with $s$ and $\otheta^{*}$ such that again $\otheta^{*} \in \oTheta^{T_s}$. 
		Since the condition on the $Counter$ variable in Step (8) ensures that the procedure eventually stops, this will guarantee the existence of an accepting run. 
		
		For Step (3) consider the structure $\str{F}_{\otheta^{*}}^{T_s} = \cal{F}_s({\otheta^*})$. For each $i$ we guess $\nTheta^{T_i}$
		to be the set of $M_\phi$-cuttings of counting types
				realized in  $\str{F}_{\otheta^{*}}^{T_s}$ by $T_i$-classes. This way $\theta^* \in \nTheta^{T_s}$. Since by
		condition (ii)(b) of Definition \ref{d:cert}
		$\str{F}_{\otheta^{*}}^{T_s} \models \phieq$ 
		it follows that every element of $\nTheta^{T_i}$ is $\phifull{T_i}$-admissible. 
		By Fact \ref{f:model-to-solution}, the system $\Gamma_{\bnTheta}(s,\theta^*)$  has an appropriate solution corresponding to $\str{F}_{\otheta^{*}}^{T_s}$. 
		
		In Step (4b) we take $\otheta^* \in \oTheta^{T_s}$ guaranteed by condition (ii)(c) of Definition \ref{d:cert}. In Step (6) take
		$\otheta' \in \oTheta^{T_s}$ guaranteed by condition  (i)(b) or (i)(c) of Definition \ref{d:cert}. Since in (7) this $\otheta'$ is substituted
		for $\otheta^*$ our loop invariant is retained.
		
		Assume now that {\tt GF$^2{+}$TG-FINSAT} accepts $\phi$. We explain how to construct a certificate $(\boTheta, \cal{F})$,
		$\boTheta=(\oTheta^{T_1}, \ldots, \oTheta^{T_k})$, $\cal{F}=(\cal{F}_1, \ldots, \cal{F}_k)$ of finite
		satisfiability of $\phi$ (cf.~Definition~\ref{d:cert}). Consider an accepting run $\pi$ of the procedure, and w.l.o.g.~assume that this run is \emph{regular}, i.e., if the
		procedure enters the loop (3)-(8) several times with the same values of $\otheta^*$ and $s$    
		it always makes the same guesses. By a \emph{run} we mean here a tree in which every node is labelled by a step number and values of all the variables,
		having a single outgoing edge from the nodes representing steps with existential  guesses (or steps without guesses) and multiple outgoing edges from the nodes representing steps with universal choices (corresponding to all possible choices).    
		
		For each $i$, let $\oTheta^{T_i}$ consists of the enriched $M_\phi$-counting types $\otheta^*$ with which  $\pi$ leaves Step (4) with $s=i$ on some of its branches. 
		Note that this in particular covers the types $\otheta^*$ guessed in Step (2). Moreover this also covers the types $\otheta'$ guessed in Step (6). To see this
		consider such $\otheta'$ and assume that the corresponding value of $s$ is $i'$. In Step (7) we make $\otheta^*=\otheta'$. Now, there are two possibilities in Step (8): either we jump to Step (3) or  accept and stop. 
		In the former case, we can choose to stay in Step (4) in $\theta^*$ and leave (4) with the values of $s$ and $\otheta^*$ unchanged. In the latter, 
		we have $Counter=0$ which means that in the current branch of $\pi$ the beginning of the loop (Step (3)) was entered twice with the same values of $\otheta^*$ and $   s$. Due to our assumption on the regularity
		of $\pi$ this means that there is a branch in $\pi$ in which $\otheta'$ is guessed with $s=i'$ at a moment when the value of $Counter >0$; on this branch  Step (3) is 
		then entered with $s=i'$ and $\otheta'$, and, as previously we can keep these values in Step (4). 
		
		Condition (i)(a) of Definition \ref{d:cert} is satisfied since every time we construct an enriched $M_\phi$-counting type it is explicitly required to be
		$\phi$-admissible. Conditions (i)(b) and (i)(c) are taken care of in Step (6). 
		
		Now we define $\cal{F}$. It is readily verified that for each enriched $M_\phi$-counting type $\otheta=(\theta, \cal{A}, \cal{B}) \in \oTheta^{T_i}$ there
		is a node in $\pi$ labelled with Step (3), in which $\nTheta$ with $\theta \in \nTheta^{T_i}$ is guessed. We set $\cal{F}_i(\otheta)$ to be the structure
		corresponding to the solution of $\Gamma_{\bnTheta}(s,\theta^*)$  from the last condition in Step (3), guaranteed by Lemma
		\ref{lem:from-solution-to-model}. This construction
		immediately ensures condition (ii) of Definition \ref{d:cert}. 
	\end{proof}

	\section{Discussion}
	In Section~\ref{sec:eq-with-equality} we showed that the finite satisfiability problem for \GFtEG{} is  \NExpTime{}-complete. 	In the previous section we showed that  the finite satisfiability problem 
	for  \GFtTRG{} is \TwoExpTime-complete.      
 Recall Remark~\ref{remark:partial} that the \TwoExpTime{} lower bound  applies already to \GFtPG.  This fragment is subsumed by both  
 \GFtTG{} and \GFtTRG{}, hence the \TwoExpTime{} lower bound transfers also to these fragments. Owing to the  reduction from
	Lemma \ref{l:reduction}  also the finite satisfiability problems for \GFtTG{} and \GFtTRG{} are \TwoExpTime-complete. 
		 Moreover, also the finite satisfiability problem for the combined fragment \GFtEG$+{\textsc{tsG}}$ is \NExpTime-complete, and for other combined variants, where at least one special symbol is required to be only transitive, transitive-reflexive or a partial order, 
		 ---\TwoExpTime-complete.

		 We also remark that essentially the same technique as in Section~\ref{sec:transitive} can be used to show that the finite satisfiability problem for \GFt{} with one transitive relation that is not restricted to guards is decidable in \TwoExpTime. As the first step one can transform the formula in such a way that  the binary guards in conjuncts of type $\cff$ and $\cfe$  have the form $\bing(x,y) \wedge \neg Txy \wedge \neg Tyx$. Inspection of our constructions shows that  when we require new witnesses for these conjuncts, we always provide them using 2-types not containing positive atoms of the form $Txy$ or $Tyx$. Hence, we can state the following corollary.
		 \begin{corollary}
		 The finite satisfiability problem for \GFt{} with one transitive relation is \TwoExpTime-complete. 
		 \end{corollary}

	Finally, we mention a few interesting open questions. First there are several closely related fragments, for which the complexity bounds of the satisfiability problem have already been established but it is not know if the same applies to the finite satisfiability problems. These include:
	\begin{enumerate}
		\item \GFt{} with {\em one-way} transitive guards. In this variant we assume that if the formulas
	 of the form $\exists v \psi(u,v)$ ($\forall v \psi(u,v)$), where $\{u, v \} = \{x,y \}$,
		employ a transitive guard,
		then it is of the form $Tuv$, i.e., the quantified variable occurs in the second position. 
		(\citeN{Kie06} shows 
		 that the satisfiability problem for \GFt{} with one-way transitive guards is \ExpSpace-complete).
\item Full \GF{} (with \emph{unbounded number of variables}) with transitive guards. (\citeN{ST04} show  that the satisfiability problem for this logic is \TwoExpTime-complete).
	\item \GFt{} with conjunction of transitive guards but without equality. In this variant we allow, e.g., guards of the form $T_1xy \wedge T_2yx$. (\citeN{Kaz06} shows that the satisfiability problem for this extension of \GFtTG{} remains \TwoExpTime-complete). 
	\item \GFt{} with transitive closure operators in guards. (\citeN{Michaliszyn09} shows \TwoExpTime-completeness of the satisfiability problem).
	\end{enumerate}
 	
 	More importantly, it would be interesting to extend the results for \GFt{} with special guards (and also for the above mentioned fragments) by additionally allowing \emph{constants} that play an important role in some application areas. So far constants have been carefully avoided in perhaps all decidable extensions of the guarded fragment with transitivity, as their presence no longer allows one to construct ramified models where pairs of distinct elements are members of at most one transitive relation. Nevertheless, we believe  that the techniques of this paper can be expanded to the case of \GFt{} with special guards and constants. The case with an unbounded number of variables might require additional insights. 
\section*{Acknowledgements}
The authors would like to thank Ian Pratt-Hartmann for his insightful comments on the proofs from the conference paper \cite{KT07} and the anonymous reviewers for their constructive criticism and valuable suggestions on earlier versions of this paper.

\bibliographystyle{ACM-Reference-Format-Journals}
\bibliography{gf2egtg-L}{}

\received{Month Year}{Month Year}{Month Year}

\end{document}